\documentclass{article}
\usepackage{amsmath,amsfonts,amsthm}
\newtheorem{lemma}{Lemma}
\newtheorem{proposition}{Proposition}
\newtheorem{assumption}{Assumption}
\newtheorem{corollary}{Corollary}
\newtheorem{theorem}{Theorem}
\usepackage{color}
\usepackage{graphicx}
\def\alert#1{{\color{red}#1}}
\def\indic{1\hspace{-2.5pt}\mathrm{l}}

\def\esssup_#1{\underset{#1}{\mathrm{ess\,sup\, }}}
\def\essinf_#1{\underset{#1}{\mathrm{ess\,inf\, }}}
\title{Hedging under multiple risk constraints\footnote{This work was
    partially funded by Electricit\'e de France. We thank Marie
    Bernhart for valuable comments on a premiminary version of this paper.}}
\author{Ying Jiao \\ ISFA, Universit\'e Claude Bernard -- Lyon I \\ \texttt{ying.jiao@univ-lyon1.fr}\and
  Olivier Klopfenstein \\ EDF R\&D\\
  \texttt{olivier.klopfenstein@edf.fr} \and Peter Tankov\\LPMA,
  Universit\'e Paris Diderot -- Paris 7\\ \texttt{tankov@math.univ-paris-diderot.fr}}
\date{}
\begin{document}
\maketitle

\begin{abstract}
Motivated by the asset-liability management of a nuclear power plant
operator, we consider the problem of finding the least expensive portfolio, which
outperforms a given set of stochastic benchmarks. For a specified loss function, the expected shortfall with
respect to each of the benchmarks weighted by this loss function must remain bounded by a given threshold. We
consider different alternative formulations of this problem in a complete
market setting, establish
the relationship between these formulations, present a general resolution
methodology via dynamic programming in a non-Markovian context and give explicit solutions in special
cases. 
\end{abstract}


\section{Introduction}
In various economic contexts, institutions hold assets
to cover future random liabilities. Banks and insurance companies are
required by the authorities to hold regulatory capital to cover the
risks they take. Pension funds face random future liabilities due to
longevity risk and the structure of the pension plans which may
involve variable annuity-type features. The problem of managing a
portfolio of assets under the condition of covering future liabilities
or benchmarks, in particular in the context of pension plans, is
commonly known as asset-liability management (ALM)
\cite{van2007optimal,detemple2008dynamic,martellini2012dynamic}. 

The present study is mainly motivated by the ALM problem of a nuclear
power plant operator. In several countries, energy companies operating nuclear power plants are
required by law to hold decomissioning funds, to cover the future
costs of decontaminating and dismantling the plants, as well as the
treatment and long-term storage of the radioactive waste. In
  France, this obligation was introduced by the Law n$^\circ$ 2006-739 of
  June 28, 2006 relative to sustainable management of radioactive
  substances and waste. According to this law, the three major players of the civil nuclear
industry in France (EDF, AREVA and CEA) must hold portfolios of assets
dedicated to future nuclear decomissioning charges \cite{rcdc.12}. The
value of these dedicated assets must be sufficient to cover the
discounted value of future liabilities. The discount rate is
determined by the operator, but the law stipulates that ``it may not
exceed the rate of return, as anticipated with a high degree of
confidence, of the hedging portfolio, managed with sufficient security
and liquidity''. Computing the highest possible discount rate is therefore
equivalent to finding the cheapest portfolio which covers
the future liabilities with a high degree of confidence. The law does
not define the notions such as ``high degree of confidence'' in a
precise manner, but it is clear that since the future liabilities are
subject to considerable uncertainty, some kind of probabilistic risk measure criterion such as shortfall probability or Value at Risk should be
used.

In other contexts, the regulator stipulates that a specific probabilistic
criterion should be used to measure the potential losses arising from not
being able to cover the liabilities. Basel II framework uses the Value at Risk to determine regulatory
capital for banks. Under the European Solvency II directive, insurance companies are
required to evaluate the amount of capital necessary to cover their
liabilities for a time horizon of one year with a probability of
$99.5\%$. Under such frameworks, companies must therefore hold enough assets
to limit the losses with respect to a random future benchmark in the
sense of a probabilistic loss measure. 

Motivated by these issues, in this paper, we consider the problem of an economic agent, who has
random liabilities payable at a finite set of future dates, with
expected shortfall constraints imposed at each payment date. We consider
two related questions: how to find the least expensive portfolio which allows to
satisfy the constraints (hedge the liabilities) at each date, and
how to determine the relationship between the probabilistic
constraints at different dates. For the latter issue, we propose three different formulations: with the European style constraints, the bound is
imposed on the expected shortfall at each date, computed at $t=0$; in the
time-consistent constraint, the bound is imposed at each payment date on the
next period's shortfall; finally, in the lookback-style constraint, the
bound is imposed on the expectation of the maximal shortfall over all dates.

In the literature, the problem of hedging a single random liability
under a probabilistic constraint has been introduced and studied mainly in the complete
market continuous-time setting in \cite{follmer.leukert.99,follmerleukert}, where
explicit solutions were obtained in the Black-Scholes model. This was
generalized to Markovian contexts in
\cite{bouchard2009stochastic} using stochastic control and viscosity
solutions and later extended to other classes of stochastic target
problems in \cite{bouchard2010optimal,bouchard2012stochastic}. A
related strain of literature deals with portfolio management under
additional constraint of outperforming a stochastic benchmark at a given date
\cite{gundel.weber.07,boyle.tian.07} or
a deterministic benchmark at all future dates
\cite{elkaroui.al.05}. The problem of hedging multiple stochastic
benchmarks at a finite set of future dates, has, however, received
little attention in the literature.

In our paper, motivated by practical concerns that the liability  constraints must be verified at pre-specified given dates, we place ourselves in the discrete-time setting. To find
the least expensive hedging portfolio, we adopt the classical complete market framework, which allows us to reformulate the
problem as the one of finding the smallest discrete-time
supermartingale satisfying a set of stochastic constraints. This can
be seen as an extension of the classical notion of Snell envelope,
which arises in the problem of superhedging an American contingent
claim --- see \cite[Section 6.5]{follmer}. The
market completeness assumption implies that in our context, the randomness of the
future liabilities is mainly determined by ``hedgeable'' risk factors,
such as the market risk, the interest rate risk, the inflation etc.,
rather than by unpredictable random events like changes in the regulatory
framework. 

Under the market completeness assumption, it is easy to find a
strategy for hedging all future liabilities almost surely. The
main difficulty of our problem is  due to the fact that the
constraints are given by expected loss functions and therefore
probabilistic in nature. Furthermore, we are concerned with general discrete-time processes where the Markovian property holds no longer. Our main contribution is to investigate the
interplay between the probabilistic constraints at different dates,
and to characterize the solution via
dynamic programming in a general context. We consider three types of different risk constraints which concern respectively the expected loss at all liability payment dates, the time-consistent conditional loss in a dynamic manner and the maximum of all loss scenarios. For
each of the three constraint styles that we use, we obtain a recursive formula for the least expensive hedging portfolio. Explicit examples are then developed for specific
loss functions.

Our paper is structured as follows. In Section
\ref{problem.sec} we introduce three different formulations of the
multi-objective hedging problem according to the three different
constraint styles. Section \ref{solution.sec} presents a non-Markovian
dynamic programming approach to the solution. Section
\ref{examples2d.sec} provides explicit examples in the case of
two payment dates, and Section \ref{examplesnd.sec} contains some
explicit results for the $n$-dimensional case. Finally, technical
lemmas are proven in the appendix. 


\section{Alternative problem formulations in the complete market
  setting}
\label{problem.sec}
We start with a filtered probability space $(\Omega,\mathcal F,\mathbb
P, \mathbb G := (\mathcal G_t)_{0\leq t \leq T})$, where $\mathcal G_0$ is a trivial
$\sigma$-field and $\mathcal G_T = \mathcal F$. On this space, we consider a financial market model
with a risk-free asset $(X^0_t)_{0\leq t \leq T}$ and $d$ risky assets
$(X^i_t)^{i=1,\dots,d}_{0\leq t \leq T}$. The risky assets are assumed to
be adapted to the filtration $\mathbb G$. Without
loss of generality, we take the risk-free asset to be a constant
process: $X^0_t \equiv 1$. 

We assume that there exists a class of admissible portfolio strategies
which does not need to be made precise at this point. To guarantee absence of arbitrage and
completeness, we make the following assumption: there exists a 
probability $\mathbb Q$ equivalent to $\mathbb P$ such that all admissible
self-financing portfolios are $\mathbb Q$-supermartingales, and 
for any $\mathbb Q$-supermartingale $(M_t)_{0\leq t \leq T}$, there
exists an admissible portfolio $(V_t)_{0\leq t \leq
  T}$, which satisfies $V_t = M_t$ for all $t\in [0,T]$. 

Given a finite sequence of deterministic moments $0=t_0<t_1<\dots <
t_n\leq T$,  we study the problem of an economic agent, who is liable
to make a series of payments $P_1,\dots,P_n$ at
dates $t_1,\dots,t_n$, where for each $i$, $P_i$ is $\mathcal
G_{t_i}$-measurable and satisfies $\mathbb E^{\mathbb Q}[|P_i|]<\infty$. This may model for example the cashflows
associated to a variable annuity insurance contract, or to a long-term
investment project. The portfolio of the agent, whose value is denoted
by $\widetilde V_t$ is effectively used to
make the payments, and therefore satisfies $\widetilde V_{t_i} = \widetilde
V_{t_i-} - P_i$. We assume that the agent has some tolerance for loss,
which means that the negative part of $\widetilde
V_{t_i-} - P_i$ must be bounded in a certain probabilistic sense (to be
made precise later in this section). 

To work with self-financing portfolios we introduce the portfolio
augmented with cumulated cash flows:
$$
V_t = \widetilde V_t + \sum_{i\geq 1: t_i \leq t } P_i
$$
as well as the benchmark process 
\begin{align}
S_t = \sum_{i\geq 1: t_i \leq t } P_i.\label{benchmark.eq}
\end{align}
The agent is therefore interested in finding the cheapest portfolio
process $(V_t)_{0\leq t\leq T}$ which outperforms, in a certain
probabilistic sense, at dates $t_1,\dots,t_n$, the benchmark process
$(S_t)_{0\leq t\leq T}$. Due to our market completeness assumption,
this is equivalent to finding the $\mathbb Q$-supermartigale
$(M_t)_{0\leq t\leq T}$ with the smallest initial value which
dominates the benchmark in a probabilistic sense at dates
$t_1,\dots,t_n$. 

Introduce a discrete filtration $\mathbb F = (\mathcal F_k)_{k=0,1,\dots,n}$
defined by $\mathcal F_k = \mathcal G_{t_k}$. For a $\mathbb G$-supermartingale
$(M_t)_{0\leq t\leq T}$, the discrete-time process
$(N_k)_{k=0,1,\dots, n}$ defined by $N_k = M_{t_k}$ is a $\mathbb
F$-supermartingale, and conversely, from a  $\mathbb
F$-supermartingale $N$ one can easily construct a $\mathbb
G$-supermartingale $M$ which coincides with $N$ at dates
$t_0,\dots,t_n$. Therefore, our problem can be reformulated in the
discrete time setting as the problem of finding a discrete-time
$\mathbb Q$-supermartingale $M$ with respect to the filtration
$\mathbb F$ with the smallest initial value, such that for
$k=1,\dots,n$, $M_k$ dominates in a probabilistic sense the
discrete-time benchmark $S_k$ (we use the same letter for discrete and
continuous-time benchmark). 
 
There are many natural ways of introducing risk tolerance into our
multiperiod hedging problem. In this paper, following
\cite{follmerleukert}, we consider that a bound is imposed on the
expected shortfall, weighted by a loss function $l: \mathbf
R\to \mathbf R$.  Throughout this paper, we will always assume that
the following condition on $l$ is satisfied. 

\begin{assumption}\label{assumption 1}
The function $l: \mathbf
R\to \mathbf R$ is convex, decreasing and bounded from
below. 
\end{assumption}

The above assumptions are natural to describe a loss function. A typical example will be a ``call'' function. In particular, when $l(x)=(-x)^+$,
we take the positive part of the loss and the situations when the liability is hedged will not be penalized. In some cases, we need stronger assumptions to get explicit results. The following Assumption \ref{assumption 2}  allows us to include widely-used utility functions. For example, when $l(x)=e^{-px}-1$ with $p>0$, then Assumption \ref{assumption 2} is satisfied.

\begin{assumption}\label{assumption 2}
The function $l: \mathbf R \to \mathbf R$ is strictly convex,
strictly decreasing, bounded from below and of class $C^1$. In addition, the
derivative $l'(x)$ satisfies the Inada's condition $\lim_{x\to -\infty} l'(x) = -\infty$ and $\lim_{x\rightarrow +\infty}l'(x)=0$.
\end{assumption}

We next describe the three different constraints for our problem as follows.

\paragraph{European-style constraint}
\begin{itemize}
\item[$(EU)$] Find the minimal value of $M_0$ such that there exits a $\mathbb Q$-supermartingale $(M_k)_{k=0}^n$ with 
\begin{align}
\mathbb E^\mathbb P[l(M_k - S_k) ] \leq \alpha_k\ \text{for}\ k=1,\dots,n.\label{consteu}
\end{align}
We denote the set of all $\mathbb Q$-supermartingales satisfying
\eqref{consteu} by $\mathcal M_{EU}$.
\end{itemize}
\paragraph{Time-consistent constraint}
\begin{itemize}
\item[$(TC)$] Find the minimal value of $M_0$ such that there exits a $\mathbb Q$-supermartingale $(M_k)_{k=0}^n$ with 
\begin{align}
\mathbb E^\mathbb P[l(M_k - S_k)|\mathcal F_{k-1} ] \leq \alpha_k,\,\, a.s. \,\ \text{for}\ k=1,\dots,n.\label{consttc}
\end{align}
We denote the set of all $\mathbb Q$-supermartingales satisfying
\eqref{consttc} by $\mathcal M_{TC}$.
\end{itemize}
The time-consistent constraint has an interesting interpretation as an
``American-style'' guarantee. 
\begin{proposition}\label{prop american}
Let $M$ be an $\mathbb F$-adapted process. Then the condition
\begin{equation}\label{constraint dyn}
\mathbb E^\mathbb P[l(M_k-S_k)|\mathcal F_{k-1}]\leq \alpha_k, \, a.s. \quad k=1,\cdots,n
\end{equation} is equivalent to: for all  $\mathbb F$-stopping times
$\tau$ and $\sigma$ taking values in $\{0,1,\dots,n\}$ such that
$\tau\leq \sigma$, 
\begin{equation}\label{constraint America}
\mathbb E^\mathbb P\left[\sum_{i=\tau+1}^\sigma \{l(M_{i}-S_{i})-\alpha_{i}\}\right]\leq 0.
\end{equation}
 \end{proposition}
\begin{proof}Let $X_k=\sum_{i=1}^kl(M_i-S_i)-\alpha_i$,
  $k=0,1,\dots,n$.  The condition
  \eqref{constraint dyn} is equivalent to the fact that the process
  $X$ is a $\mathbb P$-supermartingale. Equation \eqref{constraint
    America} then follows from Doob's theorem. Conversely, assume that
  \eqref{constraint America} holds and let $A\in \mathcal F_k$ for
  $k\in \{0,1,\dots,n-1\}$. Taking $\sigma = k+1$ and $\tau = k
  \mathbf 1_A + (k+1)\mathbf 1_{A^c}$, we get from \eqref{constraint
    America}, 
$$
\mathbb E^{\mathbb P}[(X_{k+1}-X_k)|\mathbf 1_A]\leq 0.
$$
Since $A$ and $k$ are arbitrary, this proves the supermartingale
property of $X$. 
\end{proof}

\paragraph{Lookback-style constraint}
\begin{itemize}
\item[$(LB)$] Find the minimal value of $M_0$ such that there exits a $\mathbb Q$-supermartingale $(M_k)_{k=0}^n$ with 
\begin{align}
\mathbb E^\mathbb P[\max_{k=1,\dots,n}\{l(M_k - S_k)-\alpha_k\} ] \leq 0.\label{constlb}
\end{align}
We denote the set of all $\mathbb Q$-supermartingales satisfying
\eqref{constlb} by $\mathcal M_{LB}$.
\end{itemize}

\paragraph{Relationship between different types of constraints}

For a given family of bounds $(\alpha_1,\cdots,\alpha_n)$, we denote by $V^{EU}_0(\alpha_1,\cdots,\alpha_n)$, $V^{TC}_0(\alpha_1,\cdots,\alpha_n)$ and $V^{LB}_0(\alpha_1,\cdots,\alpha_n)$ (or $V^{EU}_0$,$V^{TC}_0$ and $V^{LB}_0$ when there is no ambiguity) the  infimum value of $M_0$ where
$(M_k)_{k=0}^n$ belongs to, respectively, $\mathcal M_{EU}$, $\mathcal
M_{TC}$ and $\mathcal M_{LB}$. The following proposition shows that 
$$
V^{EU}_0\leq
V^{TC}_0\leq V^{LB}_0.
$$
\begin{proposition}\label{pro:inclusion}
The following inclusion holds true: {$\mathcal M_{LB}\subset \mathcal
M_{TC}\subset \mathcal M_{EU}$.} 
\end{proposition}
\begin{proof}
Clearly, $\mathcal M_{TC}\subset \mathcal M_{EU}$. The inclusion {$\mathcal M_{LB}\subset \mathcal
M_{TC}$} follows from the alternative representation \eqref{constraint
  America} for the time-consistent constraint. 
\end{proof}

When $n=1$, the three constraint types coincide. In addition,  the value function
has an explicit form in some particular cases. 
\begin{proposition} \label{n1.prop}Let $n=1$ and assume that $\alpha >
\lim_{x\to +\infty} l(x)$. 

\begin{enumerate}
\item 
Let Assumption \ref{assumption 2} hold true and assume that there exists $y<0$ with $\mathbb
E^{\mathbb P}[l(I(yZ))]<\infty$, where $Z=\frac{d\mathbb Q}{d\mathbb
    P}$.
Then
$$
V^{EU}_0 (\alpha) = V^{TC}_0 (\alpha)= V^{LB}_0 (\alpha)= \mathbb
E^{\mathbb Q}[S_1] + \mathbb E^{\mathbb Q}[I(\lambda^* Z)],
$$
where $I$ is the inverse function of $l'$ and $\lambda^*$ is the unique
solution of $\mathbb E^{\mathbb P}[l(I(\lambda^* Z))] = \alpha$. 
\item
Assume that $\mathbb P = \mathbb Q$. Then,
$$V^{EU}_0(\alpha)  = V^{TC}_0(\alpha)= V^{LB}_0(\alpha) = \mathbb E[S_1] + l^{-1}(\alpha)$$ where $l^{-1}(\alpha) := \inf\{x\in \mathbf R|l(x)\leq \alpha\}$. 
\end{enumerate}
\end{proposition}
\begin{proof}
The first part is a particular case of Proposition \label{prop:value
  function time consistent} which will be shown below. This result is
also very
similar to the solution of the classical concave utility maximization problem for
which we refer the reader for
example to \cite[Theorem 2.0]{kramkov1999asymptotic}. The second part is obtained by using the Jensen's inequality.
\end{proof}

The following proposition is a natural generalization of the second case of  Proposition \ref{n1.prop} and describes another situation when the three
value functions coincide. However, we need an extra hypothesis \eqref{assterm.eq} for $n\geq 2$.
\begin{proposition}
Assume  that $\mathbb P=\mathbb Q$, $\alpha_k >
\lim_{x\to +\infty} l(x)$ for $k\in \{1,\dots,n\}$  and
\begin{align}
\mathbb E[S_n + l^{-1} (\alpha_n)|\mathcal F_k] \geq S_k + l^{-1}(\alpha_k)\label{assterm.eq}
\end{align}
 for any $k\in\{1,\ldots,n\}$, where $l^{-1}(\alpha) := \inf\{x\in \mathbf R|l(x)\leq \alpha\}$.
Then,
\begin{equation*}
V^{EU}_0(\alpha_1,\dots,\alpha_n)=V^{TC}_0(\alpha_1,\dots,\alpha_n)=V^{LB}_0(\alpha_1,\dots,\alpha_n)=\mathbb E[S_n]+l^{-1}(\alpha_n).
\end{equation*}
\end{proposition}
Assumption \eqref{assterm.eq} is implied for example by the
following assumptions.
\begin{itemize}
\item The process $(S_k+l^{-1}(\alpha_k))_{k=0}^n$ is a submartingale.
\item The process $S$ is a submartingale and $\alpha_k \geq \alpha_n$
  for $0\leq k \leq n-1$. 
\end{itemize}
\begin{proof}
From Proposition \ref{n1.prop}, by removing all constraints except the
terminal one, we get
$$
V^{EU}_0(\alpha_1,\dots,\alpha_n)\geq\mathbb E[S_n]+l^{-1}(\alpha_n).
$$
To show the reverse inequality, let 
$$
M_k = \mathbb E[S_n + l^{-1}(\alpha_n)|\mathcal F_k].
$$
By assumption \eqref{assterm.eq}, 
$M_k \geq S_k + l^{-1}(\alpha_k)$, which means that $M\in \mathcal
M_{LB}$. Therefore,
$$
V^{LB}_0(\alpha_1,\dots,\alpha_n)\leq M_0 = \mathbb E[S_n]+l^{-1}(\alpha_n).
$$
The proof is completed by applying Proposition \ref{pro:inclusion}. 
\end{proof}


\section{Solution via dynamic programming}
\label{solution.sec}
In this section, we solve our optimization problem for the three types
of constraints. The main method consists of using the dynamic
programming principle  where the constraints are to be verified at
each payment date.  

\subsection{European-style constraint}

We begin by the European-style constraint which consists of a family of expectation constraints.
Denote by $V_0(\alpha_1,\ldots,\alpha_n)$ the  infimum value of $M_0$ where $(M_k)_{k=0}^n\in\mathcal M_{EU}$. The following result characterizes $V_0(\alpha_1,\ldots,\alpha_n)$ by using a family of $\mathbb P$-supermartingales, each of which corresponds to one expectation constraint.

\begin{proposition}\label{pro: first condition V0}
$V_0(\alpha_1,\ldots,\alpha_n)$ equals the infimum value of $M_0$ such that there exist a $\mathbb Q$-supermatingale $(M_t)_{t=0}^n$ and a family of $\mathbb P$-supermartingales $(N^k_t)_{t=0}^k$, $k\in\{1,\ldots,n\}$, satisfying $N_0^k=\alpha_k$ and $N_k^k\geq l(M_k-S_k)$.
\end{proposition}
\begin{proof}
Denote by $\widetilde{\mathcal M}_{EU}$ the set of all
$((M_t)_{t=0}^n,(N^k_t)_{t=0}^k,\;k\in\{1,\ldots,n\})$ which  satisfy
the conditions in the proposition and by $\widetilde V_0(\alpha_1,\ldots,\alpha_n)$ the infimum value of the related $M_0$. If $((M_t)_{t=0}^n,(N^k_t)_{t=0,\cdots, k}^{k=1,\cdots,n})\in\widetilde{\mathcal M}_{EU}$, then  
\[\alpha_k=N_0^k\geq\mathbb E^{\mathbb P}[N_k^k]\geq\mathbb E^{\mathbb P}[l(M_k-S_k)].\]
Therefore $(M_k)_{k=0}^n\in\mathcal M_{EU}$, which implies $\widetilde
V_0(\alpha_1,\ldots,\alpha_n)\geq
V_0(\alpha_1,\ldots,\alpha_n)$. Conversely, given a $\mathbb
Q$-supermartingale $(M_k)_{k=0}^n\in\mathcal M_{EU}$, we construct
processes $(N^k_t)_{t=0}^k$ such that $N_0^k=\alpha_k$ and
$N_t^k=\mathbb E[l(M_k-S_k)|\mathcal F_t]$ for
$t\in\{1,\ldots,k\}$. The condition $\mathbb E^{\mathbb
  P}[l(M_k-S_k)]\leq\alpha_k$ implies that $(N^k_t)_{t=0}^k$ is a
$\mathbb P$-supermartingale. It follows that the inverse inequality
$\widetilde V_0(\alpha_1,\ldots,\alpha_n)\leq
V_0(\alpha_1,\ldots,\alpha_n)$ holds true as well.
\end{proof}

For the dynamic solution, we will introduce the value-function process
$V_k$ for $k\in\{1,\ldots,n\}$. Note that, instead of a family of
real-valued bounds $(\alpha_1,\cdots,\alpha_n)$ for $V_0$, the value
function $V_k$ will take as parameters a family of random variables
which describe the evolution of the successive bounds on
the expected shortfall. 

For $\mathcal F_k$-measurable random variables $N^{k+1},\ldots,N^n$, let $V_k(N^{k+1},\ldots,N^n)$ be the essential infimum of all  $M_k\in \mathcal F_k$ such that
there exists a $\mathbb Q$-supermartingale $(M_t)_{t=k}^n$ with $\mathbb E^{\mathbb P}[l(M_t-S_t)|\mathcal F_k]\leq N^t$ for any $t\in\{k+1,\ldots,n\}$. We denote by $\mathcal M_{EU,k}(N^{k+1},\ldots,N^n)$ the set of all such $\mathbb Q$-supermartingales which start from $k$.
By convention, let $V_n=-\infty$ and $\mathcal M_{EU,n}$ coincide with the set of all $\mathcal F_n$-measurable random variables.  We have in particular
\begin{align*}
V_{n-1}(N^n) &= \text{essinf}_{M_n\in\mathcal F_n}\{\mathbb E^{\mathbb Q}[M_n|\mathcal F_{n-1}]: \mathbb E^{\mathbb P}[l(M_n-S_n)|\mathcal F_{n-1}]\leq N^n\}\\
& = \mathbb E^{\mathbb Q}[S_n|\mathcal F_{n-1}] + \text{essinf}_{M\in\mathcal
  F_n}\{\mathbb E^{\mathbb Q}[M|\mathcal F_{n-1}]: \mathbb E^{\mathbb P}[l(M)|\mathcal
F_{n-1}]\leq N^n\}.
\end{align*}
In fact, if $(M_{n-1},M_n)$ is a $\mathbb Q$-supermartingale such that $\mathbb E^{\mathbb P}[l(M_n-S_n)|\mathcal F_{n-1}]\leq N^n$, then $(\mathbb E^{\mathbb Q}[M_n|\mathcal F_{n-1}],M_n)$ is also a $\mathbb Q$-martingale satisfying this condition and one has $M_{n-1}\geq\mathbb E^{\mathbb Q}[M_n|\mathcal F_{n-1}]$. 

Similarly to Proposition \ref{pro: first condition V0}, we can prove that $V_k(N^{k+1},\ldots,N^n)$ equals the essential infimum of $M_k\in \mathcal F_k$ such that there exists a $\mathbb Q$-supermartingale $(M_t)_{t=k}^n$ together with a family of $\mathbb P$-supermartingales $(N^{k+1}_t)_{t=k}^{k+1},\cdots, (N^{n}_t)_{t=k}^{n} $ such that $N_k^t=N^t$ and that $N_t^t\geq l(M_t-S_t)$ for any $t\in\{k+1,\ldots,n\}$. 

\begin{lemma}\label{Lem:filtrant}
Let $(M_t)_{t=k}^n,(M'_t)_{t=k}^n\in\mathcal M_{EU,k}(N^{k+1},\ldots,N^n)$, then there exists a supermartingale $(M_t'')_{t=k}^n\in\mathcal M_{EU,k}(N^{k+1},\ldots,N^n)$ such that $M_k''=M_k\wedge M_k'$.
\end{lemma}
\begin{proof}
By definition, $\mathbb E^{\mathbb P}[l(M_t-S_t)|\mathcal F_k]\leq N^t$ and $\mathbb E^{\mathbb P}[l(M_t'-S_t)|\mathcal F_k]\leq N^t$ for any $t\in\{k+1,\ldots,n\}$. Let $A$ be the set $\{M_k\geq M_k'\}$, which is $\mathcal F_k$-measurable. For any $t\in\{k,\ldots,n\}$, let $M_t''=\indic_AM_t'+\indic_{A^c}M_t$. Since $A\in \mathcal F_k$,  $(M_t'')_{t=k}^n$ is a $\mathbb Q$-supermartingale. Moreover, $M_k''=\min(M_k,M_k')$ and
\[\mathbb E^{\mathbb P}[l(M_t''-S_t)|\mathcal F_k]=
\indic_A\mathbb E^{\mathbb P}[l(M_t'-S_t)|\mathcal F_k]+\indic_{A^c}\mathbb E^{\mathbb P}[l(M_t-S_t)|\mathcal F_k]\leq N^t,\]
therefore $(M_t'')_{t=k}^n\in\mathcal M_{EU,k}(N^{k+1},\ldots,N^n)$.
\end{proof}

The following result characterizes the value-function process in a backward and recursive form.

\begin{theorem}\label{EU main result}
$V_k(N^{k+1},\ldots,N^n)$ equals the essential infimum of all $M_k\in\mathcal F_k$ such that there exist  $M_{k+1}\in\mathcal F_{k+1}$ and a family of $\mathbb P$-supermartingales $(N^{k+1}_t)_{t=k}^{k+1}$, $\cdots$, $(N^{n}_t)_{t=k}^{n} $ which satisfy the following conditions~: 
\begin{equation}\label{Equ:stoprog}\begin{cases}
\mathbb E^{\mathbb Q}[M_{k+1}|\mathcal F_k]=M_k,\\ 
N_k^i=N^i \text{ for  $i\in\{k+1,\ldots,n\}$},\\ 
N_{k+1}^{k+1}\geq l(M_{k+1}-S_{k+1}),\\
M_{k+1}\geq V_{k+1}(N_{k+1}^{k+2},\ldots,N_{k+1}^n).
\end{cases}\end{equation}
\end{theorem}
\begin{proof} By  definition,  the theorem is true for $k=n-1$. In the following, we treat the case $k\in\{0,\ldots,n-2\}$ by induction.
Denote by $\widetilde V_k(N^{k+1},\ldots,N^n)$ the essential infimum defined in the theorem. Let $M_{k+1}$ be an $\mathcal F_{k+1}$-measurable random variable and $(N^{k+1}_t)_{t=k}^{k+1}$, $\cdots$, $(N^{n}_t)_{t=k}^{n} $ be $\mathbb P$-supermartingales which verify the conditions in \eqref{Equ:stoprog}. Note that the fourth condition, together with the essential infimum property, imply (e.g. \cite[V.18]{doob}) the existence of a sequence of random variables $(M_{k+1}^{(m)})_{m\in\mathbb N}$ in $\mathcal M_{EU,k+1}(N_{k+1}^{k+2},\ldots,N_{k+1}^n)$ such that $M_{k+1}\geq\inf_{m\in\mathbb N}M_{k+1}^{(m)}$. Moreover, the previous lemma shows that the sequence $(M_{k+1}^{(m)})_{m\in\mathbb N}$ can be supposed decreasing without loss of generality.
By  a similar argument as in Proposition \ref{pro: first condition V0}, there exist
 $\mathbb Q$-supermartingales $(M_t^{(m)})_{t=k+1}^n$  and $\mathbb P$-supermartingales $(N^{k+2,(m)}_t)_{t=k+1}^{k+2}$, $\cdots$, $(N^{n, (m)}_t)_{t=k+1}^{n} $ such that \[N_{k+1}^{i,(m)}=N_{k+1}^i \text{ and }N_{i}^{i,(m)}\geq l(M_i^{(m)}-S_i)\text{ for }i\in\{k+1,\ldots,n\}.\]
We extend each $(N^{i,(m)}_t)_{t=k+1}^i$ to a $\mathbb P$-supermartingale on $t=k,\ldots, i$ by taking $N_k^{i,(m)}=N^i$, $i\in\{ k+2,\ldots,n\}$. Let $M_k^{(m)}=\mathbb E^{\mathbb Q}[M_{k+1}^{(m)}|\mathcal F_k]$ for any $m\in\mathbb N$. Note that the $\mathbb Q$-supermartingale $(M_t^{(m)})_{t=k}^n$ and the $\mathbb P$-supermartingales $(N^{k+1}_t)_{t=k}^{k+1}$, $(N^{k+2,(m)}_t)_{t=k}^{k+2},\ldots,(N^{n,(m)}_t)_{t=k}^n$ verify the conditions which characterize $V_k(N^{k+1},\ldots,N^n)$. Therefore $V_k(N^{k+1},\ldots,N^n)\leq M_{k}^{(m)}$ for any $m\in\mathbb N$, which implies that
 \[V_{k}(N^{k+1},\ldots,N^n)\leq\inf_{m\in\mathbb N}\mathbb E^{\mathbb Q}[M_{k+1}^{(m)}|\mathcal F_k]\leq
 \mathbb E^{\mathbb Q}[M_{k+1}|\mathcal F_t],\]
where the second inequality comes from the relation $M_{k+1}\geq\inf_{m\in\mathbb N}M_{k+1}^{(m)}$ and the fact that the sequence $(M_{k+1}^{(m)})_{m\in\mathbb N}$ is decreasing. Hence we obtain 
$V_k(N^{k+1},\ldots,N^n)\leq \widetilde V_{k}(y_{k+1},\ldots,y_n)$.

Conversely, let $(M_t)_{t=k}^n$  be a $\mathbb Q$-supermartingale and $(N^{k+1}_t)_{t=k}^{k+1},\ldots,(N^n_t)_{t=k}^n$ be $\mathbb P$-supermartingales such that $N_k^i=N^i$ and that $N_i^i\geq l(M_i-S_i)$ for any $i\in\{k+1,\ldots,n\}$. Without loss of generality, we assume  $M_k=\mathbb E^{\mathbb Q}[M_{k+1}|\mathcal F_k]$. The conditions in \eqref{Equ:stoprog} are easily verified for $M_{k+1}$, $(N^{k+1}_t)_{t=k}^{k+1},\ldots,(N^n_t)_{t=k}^n$ except the last one. Note that for any $i\in\{k+1,\ldots,n\}$, the condition $N_i^i\geq l(M_i-S_i)$ implies \[N_{k+1}^i\geq\mathbb E^{\mathbb P}[N_i^i|\mathcal F_{k+1}]\geq\mathbb E^{\mathbb P}[l(M_i-S_i)|\mathcal F_{k+1}],\]
where the first inequality comes from the fact that $(N_t^i)_{t=k}^i$ is a $\mathbb P$-supermartingale.
Therefore, $M_{k+1}\geq V_{k+1}(N_{k+1}^{k+2},\ldots,N_{k+1}^n)$, which leads to the inverse inequality.
\end{proof} 

The supermartingale condition in the previous theorem can be made simpler where only $\mathcal F_{k+1}$-measurable random variables are concerned. In particular, we consider the  case of real-valued  bounds.

\begin{corollary}\label{Cor:dynamic}
$V_k(\alpha_{k+1},\ldots,\alpha_n)$ equals the essential infimum of all conditional expectations $\mathbb E^{\mathbb Q}[M|\mathcal F_k]$, $M\in\mathcal F_{k+1}$ such that there exist $\mathcal F_{k+1}$-measurable random variables $N^{k+2},\cdots, N^n$ satisfying the following conditions:
\begin{equation}\label{Equ:europfinale}
\begin{cases}
\mathbb E^{\mathbb P}[N^i|\mathcal F_k]=\alpha_i \text{ for }i\in\{k+2,\ldots,n\},\\
\mathbb E^{\mathbb P}[l(M-S_{k+1})|\mathcal F_k]\leq \alpha_{k+1},\\
M\geq V_{k+1}(N^{k+2},\ldots,N^n).\\
\end{cases}
\end{equation} 
\end{corollary}
\begin{proof}
Denote by $\widetilde V_k(\alpha_{k+1},\ldots,\alpha_n)$ the essential infimum defined in the corollary. We begin with $\mathcal F_{k+1}$-measurable random variables $M$, $N^{k+2},\ldots,N^n$ satisfying the conditions in \eqref{Equ:europfinale}. Let $M_{k+1}=M$ and $M_k=\mathbb E^{\mathbb Q}[M|\mathcal F_k]$. We then construct $\mathbb P$-supermartingales $(N^{k+1}_t)_{t=k}^{k+1},\ldots,(N^n_t)_{t=k}^n$ with $N_k^i=\alpha_i$ for $i\in\{k+1,\ldots,n\}$, $N_{k+1}^{k+1}=l(M-S_{k+1})$ and $N_{k+1}^i=N^i$ for $i\in\{k+2,\ldots,n\}$. Therefore, $V_k(\alpha_{k+1},\ldots,\alpha_n)\leq\widetilde V_k(\alpha_{k+1},\ldots,\alpha_n)$.

Note that the essential infimum $\widetilde V_k(\alpha_{k+1},\ldots,\alpha_n)$ remains unchanged if we replace the first condition of \eqref{Equ:europfinale} by ``$\mathbb E^{\mathbb P}[N^i|\mathcal F_k]\leq \alpha_i$ \emph{for any} $i\in\{k+2,\ldots,n\}$'' since the function $V_{k+1}$ is descreasing with respect to each coordinate. Let $M_{k+1}\in\mathcal F_{k+1}$ and $(N^{k+1}_t)_{t=k}^{k+1},\ldots,(N^n_t)_{t=k}^n$ satisfy the conditions in \eqref{Equ:stoprog}. We choose $M=M_{k+1}$ and $N^i=N_{k+1}^i$ for $i\in\{k+2,\ldots,n\}$. By \eqref{Equ:stoprog}, $M\geq V_{k+1}(N^{k+2},\ldots,N^n)$ and $\mathbb E^{\mathbb P}[N^i|\mathcal F_k]\leq N_k^i=\alpha_i$. Moreover, since $N_{k+1}^{k+1}\geq l(M_{k+1}-S_{k+1})$, one has 
\[\alpha_{k+1}=N_{k}^{k+1}\geq \mathbb E^{\mathbb P}[N_{k+1}^{k+1}|\mathcal F_k]\geq\mathbb E^{\mathbb P}[l(M-S_{k+1})|\mathcal F_k].\]
Hence $\widetilde V_k(\alpha_{k+1},\ldots,\alpha_n)\leq V_k(\alpha_{k+1},\ldots,\alpha_n)$.

\end{proof}

\subsection{Time-consistent constraint}

For the time-consistent constraint, we consider the loss function at each payment date given the market information at the previous time step. The constraint is written by using conditional expectations. The dynamic programming structure in this setting  is relatively simple since it only involves two successive dates.  

Recall that $\mathcal M_{TC}$ is the set of all $\mathbb Q$-supermartingales $(M_k)_{k=0}^n$ with 
\begin{align*}
\mathbb E^\mathbb P[l(M_k - S_k)|\mathcal F_{k-1} ] \leq \alpha_k\ \text{for}\ k=1,\dots,n.
\end{align*}
In a dynamic manner,
denote by $\mathcal M_{TC,k}$ the set of all  $\mathbb Q$-supermartingales $(M_t)_{t=k}^n$ verifying the condition
\[\mathbb E^{\mathbb P}[l(M_t-S_t)|\mathcal F_{t-1}]\leq\alpha_t\text{ for }t=k+1,\ldots,n.\]
Note that $\mathcal M_{TC,n}$ is the set of all integrable $\mathcal F_n$-measurable random variables and $\mathcal M_{TC,0}$ coincides with $\mathcal M_{TC}$.

Define the value-function process $(V_k)_{0\leq k \leq n}$ in a backward manner as follows~:  $V_n=-\infty$, and for $k\in\{1,\ldots,n\}$, let
\begin{equation}\label{V k-1 americaine}
V_{k-1} = \essinf_{M \in \mathcal F_k}\{ \mathbb E^\mathbb Q[M|\mathcal F_{k-1}]: M\geq V_k,\text{ and } \mathbb E^\mathbb P[l(M-S_k)|\mathcal
  F_{k-1}] \leq \alpha_k\} 
\end{equation}

\begin{proposition}\label{Pro:TCprogdy} For any $k\in\{0,\ldots,n\}$, 
$$
V_k = \essinf_{M_k\in \mathcal M_{TC,k}} M_k. 
$$
\end{proposition}
\begin{proof} The case where $k=n$ is trivial. In the following, we assume that the equality $V_k = \essinf_{M_k\in \mathcal M_{TC,k}} M_k$ has been proved.

Let $M_{k-1}$ be an element in $\mathcal M_{TC,k-1}$. There then exists a $\mathbb Q$-supermartingale $(M_t)_{t=k-1}^n$ such that 
\[\mathbb E^{\mathbb P}[l(M_t-S_t)|\mathcal F_{t-1}]\leq\alpha_t\text{ for }t=k,\ldots,n,\]
which implies that $M_{k}\in\mathcal M_{TC,k}$. Therefore  $M_k\geq V_k$ by the induction hypothesis. Combining this with the condition $\mathbb E^{\mathbb P}[l(M_k-S_k)|\mathcal F_{k-1}]$, we obtain by definition \eqref{V k-1 americaine} that $V_{k-1}\leq\mathbb E^{\mathbb Q}[M_k|\mathcal F_{k-1}]\leq M_{k-1}$. Hence $V_{k-1}\leq\essinf_{} \mathcal M_{TC,k-1}$.

For the converse inequality, let $M\in\mathcal F_k$ such that $M\geq V_k$ and $\mathbb E^{\mathbb P}[l(M-S_k)|\mathcal F_{k-1}]\leq\alpha_k$. By the induction hypothesis and an argument similar to Lemma \ref{Lem:filtrant}, 
there exists a family $(M_t^{(m)})_{t=k}^n$ of $\mathbb Q$-supermartingales indexed by $m\in\mathbb N$ such that 
\[\mathbb E^{\mathbb P}[l(M_t^{(m)}-S_t)|\mathcal F_{t-1}]\leq\alpha_t\text{ for }t=k+1,\ldots,n,\]
and  $\mathcal F_k$-measurable sequence $(M_k^{(m)})_{m\in\mathbb N}$  is decreasing and  $M\geq\inf_{m\in\mathbb N}M_k^{(m)}$. For each $m\in\mathbb N$, let $\widetilde M_{k}^{(m)}=\max(M_k^{(m)},M)$, $\widetilde M_{k-1}^{(m)}=\mathbb E^{\mathbb Q}[\widetilde M_k^{(m)}|\mathcal F_{k-1}]$ and $\widetilde M_t^{(m)}=M_t^{(m)}$ for $t\in\{k+1,\ldots,n\}$. Then $(\widetilde M_t^{(m)})_{t=k-1}^n$ is a $\mathbb Q$-supermartingale. Moreover, \[\mathbb E^{\mathbb P}[l(\widetilde M_t^{(m)}-S_t)|\mathcal F_{t-1}]\leq\alpha_t\text{ for }t=k,\ldots,n,\]
which implies $\widetilde M_{k-1}^{(m)}\in\mathcal M_{TC,k-1}$. Therefore 
\[\essinf_{}\mathcal M_{TC,k-1}\leq\widetilde M_{k-1}^{(m)}=\mathbb E^{\mathbb Q}[M_k^{(m)}\vee M|\mathcal F_{k-1}]\]
for any $m\in\mathbb N$. By taking the limit when $m$ goes to the infinity, one obtains $\essinf_{}\mathcal M_{TC,k-1}\leq\mathbb E^{\mathbb Q}[M|\mathcal F_{k-1}]$ by  $M\geq\inf_{m\in\mathbb N}M_k^{(m)}$. Since $M$ is arbitrary, $\essinf_{}\mathcal M_{TC,k-1}\leq V_{k-1}$. The result is thus proved.
\end{proof}

By the previous proposition, we can calculate the smallest initial
capital $V_0$ by using the recursive formula for the value function
\eqref{V k-1 americaine}. With extra regularity condition on the loss
function $l$, we can obtain a more explicit result as follows.    

\begin{proposition}\label{prop:value function time consistent}
Let $l$ be a loss function satisfying  Assumption \ref{assumption 2} and let $I:(-\infty,0)\rightarrow \mathbf R$ be the inverse function of $l'$. Suppose in addition that for any $k\in\{1,\ldots,n\}$,
$\alpha_k > \lim_{x\to +\infty} l(x)$ and  there exists a strictly negative random variable $Y_{k-1}\in\mathcal F_{k-1}$ such that
\begin{equation}\label{Equ:dominante}\mathbb E^{\mathbb P}[l(I(Y_{k-1}Z_k/Z_{k-1}))]<+\infty\,\text{ where } Z_k = \mathbb E^\mathbb P[\frac{d\mathbb Q}{d\mathbb P}\,|\, \mathcal F_k].\end{equation} 
Then the value
function $V_t$ satisfies $V_t = \widehat V_t$ a.s. for $t=0,\dots,n-1$
where $\widehat V_{t}$ is given by
\begin{align*}
\widehat V_{n-1} &= \mathbb E^\mathbb Q[S_n + I(\lambda_{n-1} Z_n/Z_{n-1})|\mathcal
F_{n-1}]
\end{align*}
where $\lambda_{n-1} \in \mathcal F_{n-1}$ is the solution of
\begin{align}
\mathbb E^\mathbb P[l(I(\lambda_{n-1} Z_n/Z_{n-1}))|\mathcal F_{n-1}] = \alpha_n,\label{lambdastep1}
\end{align}
and for $k<n$,
\begin{multline*}
\widehat V_{k-1}= \mathbb E^\mathbb Q[\widehat V_k|\mathcal F_{k-1}] \\+ \mathbb
E^\mathbb Q\left[\left\{S_k - \widehat V_k+
    I(\lambda_{k-1} Z_k/Z_{k-1})\right\}^+\Big| \mathcal F_{k-1}\right]\mathbf
1_{\mathbb E^\mathbb P[l(\widehat V_k-S_k)|\mathcal F_{k-1}]> \alpha_k},
\end{multline*}
where $\lambda_{k-1}\in \mathcal F_{k-1}$ is the solution of 
\begin{align}\label{Equ:lambdak-1}
\mathbb E^\mathbb P\left[l\left(I(\lambda_{k-1} Z_k/Z_{k-1})\vee (\widehat V_k -
    S_k)\right)\Big| \mathcal F_{k-1}\right] = \alpha_k.
\end{align}

\end{proposition}
\begin{proof}${}$\\
\noindent\textit{Step 1. Existence of $\lambda$.}\ For any $k\in\{1,\ldots,n\}$, let $\Lambda_{k-1}$ be the set of strictly negative $\mathcal F_{k-1}$-measurable random variables $Y$ such that
\begin{align*}
\mathbb E^\mathbb P\left[l\left(I(Y Z_k/Z_{k-1})\vee (\widehat V_k -
    S_k)\right)\Big| \mathcal F_{k-1}\right] \leq \alpha_k.
\end{align*} By using the dominated convergence theorem, the condition $\alpha_k>\lim_{x\rightarrow+\infty}l(x)$ and assumption \eqref{Equ:dominante} imply that the family $\Lambda_{k-1}$ is not empty.
Let $\lambda_{k-1}$ be the essential infimum of this family. Note that $\Lambda_{k-1}$ is stable by taking the infimum of finitely many random variables. Therefore $\lambda_{k-1}$ can be written as the limit of a decreasing sequence in $\Lambda_{k-1}$. The continuity of the functions $I$ and $l$, together with the monotone convergence theorem, show that $\lambda_{k-1}$ lies in the family $\Lambda_{k-1}$. It remains to show the equality \eqref{Equ:lambdak-1}. If the equality does not hold, then for sufficiently small $\varepsilon>0$, the set $A_\varepsilon$ of $\omega\in\Omega$ such that 
\[\mathbb E^\mathbb P\left[l\left(I((\lambda_{k-1}-\varepsilon) Z_k/Z_{k-1})\vee (\widehat V_k -
    S_k)\right)\Big| \mathcal F_{k-1}\right](\omega) \leq \alpha_k\]
has a strictly postive measure. This implies that $\lambda_{k-1}-\varepsilon\indic_{A_{\varepsilon}}$ also lies in $\Lambda_{k-1}$, which leads to a contradiction.

\noindent\textit{Step 2. Representation for $V_{n-1}$.}\ By Proposition \ref{Pro:TCprogdy},
$$V_{n-1} = \mathbb E^{\mathbb Q}[S_n|\mathcal
F_{n-1}]+\essinf_{M\in
\mathcal F_n}\{\mathbb E^{\mathbb Q}[ M|\mathcal F_{n-1}]: \mathbb E^{\mathbb P}[l(M)|\mathcal F_{n-1}]\leq
\alpha_n\}$$
Taking $M = I(\lambda_{n-1} Z_n/Z_{n-1})$, we see that
$V_{n-1}\leq \widehat V_{n-1}$. On the other hand, for every strictly negative random variable
$\lambda\in \mathcal F_{n-1}$, $l(M)\geq l^*(\lambda Z_n/Z_{n-1}) + \lambda M
Z_n/Z_{n-1}$ where $l^*(u) = \inf_{v}\{ l(v) - uv\}$ is the Legendre transformation of $l$. Note that one has $(l^*)'=I$. So 
\begin{align*}
&V_{n-1} -  \mathbb E^{\mathbb Q}[S_n|\mathcal
F_{n-1}]\\
&\geq \essinf_{M\in
\mathcal F_n}\{\mathbb E^{\mathbb Q}[ M|\mathcal F_{n-1}]:  \mathbb E^{\mathbb P}[l^*(\lambda Z_n/Z_{n-1}) + \lambda M Z_n/Z_{n-1}|\mathcal F_{n-1}]\leq
\alpha_n\}\\
& = \essinf_{M\in
\mathcal F_n}\{\mathbb E^{\mathbb Q}[ M|\mathcal F_{n-1}]:  \mathbb E^{\mathbb P}[l^*(\lambda Z_n/Z_{n-1})|\mathcal F_{n-1} ]+\lambda
\mathbb
E^{\mathbb Q}[M |\mathcal F_{n-1}]\leq
\alpha_n\}\\
& = \frac{\alpha_n -\mathbb E^{\mathbb P}[l^*(\lambda Z_n/Z_{n-1})
  |\mathcal F_{n-1}]
}{\lambda}\\
&=\frac{\alpha_n -\mathbb E^{\mathbb P}[l(I(\lambda
  Z_n/Z_{n-1})) |\mathcal F_{n-1}]
}{\lambda} + \mathbb E^{\mathbb Q}[
I(\lambda
  Z_n/Z_{n-1}) |\mathcal F_{n-1}],
\end{align*}
where the last equality comes from the relation $l^*(y)=l(I(y))-yI(y)$. 
Taking $\lambda$ to be the solution of \eqref{lambdastep1}, we find
that $V_{n-1}\geq \widehat V_{n-1}$, which proves the desired
representation of $V_{n-1}$. 

\noindent\textit{Step 3. General case.}\ We now proceed to the proof of the general case by induction on $k$. Assume that we have already established the equality $V_{k}=\widehat{V}_k$. By Proposition \ref{Pro:TCprogdy} and this induction hypothesis, 
\[V_{k-1}=\essinf_{M\in\mathcal F_k}\{\mathbb E^{\mathbb Q}[M|\mathcal F_{k-1}]\,:\,M\geq\widehat V_k,\;
\mathbb E^{\mathbb P}[l(M-S_k)|\mathcal F_{k-1}]\leq\alpha_k\}.\]
We choose
\[\widetilde M:=\widehat V_k+\big(S_k - \widehat V_k+
    I(\lambda_{k-1} Z_k/Z_{k-1})\big)^+\mathbf
1_{\mathbb E^\mathbb P[l(\widehat V_k-S_k)|\mathcal F_{k-1}]> \alpha_k},\]
which is bounded from below by $\widehat{V}_k$ and satisfies $\mathbb E^{\mathbb P}[l(\widetilde M-S_k)|\mathcal F_{k-1}]\leq\alpha_k$. In fact, it is clear that this inequality holds on the set $\{\mathbb E^\mathbb P[l(\widehat V_k-S_k)|\mathcal F_{k-1}]\leq \alpha_k\}$. Moreover, on the set $\{\mathbb E^\mathbb P[l(\widehat V_k-S_k)|\mathcal F_{k-1}]> \alpha_k\}$, one has $\widetilde M-S_k=\widehat{V}_k-S_k+(S_k-\widehat V_k+I(\lambda_{k-1}Z_k/Z_{k-1}))^+$, so 
\[\mathbb E^{\mathbb P}[l(\widetilde M-S_k)|\mathcal F_{k-1}]=\mathbb E^{\mathbb P}[l(I(\lambda_{k-1}Z_k/Z_{k-1})\vee (\widehat{V}_k-S_k))|\mathcal F_{k-1}]\]which implies 
$\mathbb E^{\mathbb P}[l(\widetilde M-S_k)|\mathcal F_{k-1}]\leq\alpha_k$ on this set by the definition of $\lambda_{k-1}$. Hence 
\[V_{k-1}\leq\mathbb E^{\mathbb Q}[\widetilde M|\mathcal F_{k-1}]=\widehat{V}_{k-1}.\]

The opposite inequality $\widehat{V}_{k-1}\leq V_{k-1}$ is straightforward on the set 
$\{\mathbb E^\mathbb P[l(\widehat V_k-S_k)|\mathcal F_{k-1}]\leq \alpha_k\}$.
It remains to establish the inequality  on the set
$\{\mathbb E^\mathbb P[l(\widehat V_k-S_k)|\mathcal F_{k-1}]> \alpha_k\}$.
We have by a variable change
\begin{equation*}\begin{split}&V_{k-1}-\mathbb E^{\mathbb Q}[\widehat{V}_k|\mathcal F_{k-1}]\\
&=\essinf_{}\{\mathbb E^{\mathbb Q}[M|\mathcal F_{k-1}]:M\geq 0,\;
\mathbb E^{\mathbb P}[l(M+\widehat{V}_k-S_k)|\mathcal F_{k-1}]\leq\alpha_k\},\\
&\geq\essinf_{M\in\mathcal F_k,M\geq 0}\{\mathbb E^{\mathbb
  Q}[M|\mathcal F_{k-1}]:\mathbb E^{\mathbb P}[l^*\Big(\lambda
{Z_{k}\over Z_{k-1}}\Big)+\lambda(M+\widehat{V}_k-S_k) {Z_{k}\over Z_{k-1}}\Big|\mathcal F_{k-1}]\leq\alpha_k\}\\
&=\essinf_{M\in\mathcal F_k,M\geq 0}\{\mathbb E^{\mathbb Q}[M|\mathcal F_{k-1}]:\mathbb E^{\mathbb P}[l^*\Big(\lambda {Z_{k}\over Z_{k-1}}\Big)|\mathcal F_{k-1}]+\lambda\mathbb E^{\mathbb Q}[M+\widehat{V}_k-S_k|\mathcal F_{k-1}]\leq\alpha_k\}\\
&=\essinf_{M\in\mathcal F_k,M\geq 0}\Big\{\mathbb E^{\mathbb Q}[M|\mathcal F_{k-1}]:\\&\qquad
\mathbb E^{\mathbb Q}[M|\mathcal F_{k-1}]\geq\frac{\alpha_k-\mathbb
  E^{\mathbb P}[l(I(\lambda {Z_k\over Z_{k-1}}))|\mathcal F_{k-1}]}{\lambda}
+\mathbb E^{\mathbb Q}[S_k-\widehat{V}_k+ I\Big(\lambda {Z_k\over Z_{k-1}}\Big)\Big|\mathcal F_{k-1}]\Big\}
\end{split}
\end{equation*}
for any strictly negative $\lambda\in\mathcal F_{k-1}$. On the set 
$\{\mathbb E^\mathbb P[l(\widehat V_k-S_k)|\mathcal F_{k-1}]> \alpha_k\}$,
the last constraint above can be replaced by
\begin{multline*}\mathbb E^{\mathbb Q}[M|\mathcal
  F_{k-1}]\geq\frac{\alpha_k-\mathbb E^{\mathbb P}[l(I(\lambda
    Z_k/Z_{k-1})\vee(\widehat{V}_k-S_k))|\mathcal
    F_{k-1}]}{\lambda}\\+\mathbb E^{\mathbb Q}[\big(S_k-\widehat{V}_k+
  I(\lambda Z_k/Z_{k-1})\big)+|\mathcal F_{k-1}].
\end{multline*}
Now we choose $\lambda$ to be the solution of \eqref{Equ:lambdak-1} to obtain the following constraint
\[\mathbb E^{\mathbb Q}[M|\mathcal F_{k-1}]\geq \mathbb E^{\mathbb Q}[\big(S_k-\widehat{V}_k+I(\lambda_{k-1}Z_k/Z_{k-1})\big)+|\mathcal F_{k-1}].\]
Therefore on the set 
$\{\mathbb E^\mathbb P[l(\widehat V_k-S_k)|\mathcal F_{k-1}]> \alpha_k\}$, we have
\[V_{k-1}-\mathbb E^{\mathbb Q}[\widehat{V}_k|\mathcal F_{k-1}]\geq \mathbb E^{\mathbb Q}[\big(S_k-\widehat{V}_k+I(\lambda_{k-1}Z_k/Z_{k-1})\big)^+|\mathcal F_{k-1}],\]
which implies the inequality $\widehat{V}_{k-1}\leq V_{k-1}$.  The theorem is thus proved.
\end{proof}

In the risk-neutral case where $\mathbb P=\mathbb Q$, we can relax
Assumption \ref{assumption 2} and obtain a similar  result under
Assumption \ref{assumption 1} only.

\begin{proposition}\label{time_consistent_riskneutral:propo}
We assume that $\mathbb P = \mathbb Q$ and that for any $k\in\{1,\ldots,n\}$,
$\alpha_k > \lim_{x\to +\infty} l(x)$. 
\begin{enumerate}
\item The value function satisfies $V_t = \widehat V_t$ a.s. for $t=0,\dots,n-1$ where 
\begin{align*}
 \widehat V_{n-1} &= \mathbb E[S_n|\mathcal F_{n-1}] + l^{-1}(\alpha_n)\\
 \widehat V_{k-1} & = \mathbb E[ \widehat V_k|\mathcal F_{k-1}] \\ &+ \mathbb
E\left[\left\{S_k -  \widehat V_k+
    \lambda_k\right\}^+\Big| \mathcal F_{k-1}\right]\mathbf
1_{\mathbb E[l( \widehat V_k-S_k)|\mathcal F_{k-1}]> \alpha_k},\, k<n,
\end{align*}
where $\lambda_k\in \mathcal F_{k-1}$ is the solution of 
\begin{align}
\mathbb E\left[l\left(\lambda_k\vee ( \widehat V_k -
    S_k)\right)\Big| \mathcal F_{k-1}\right] = \alpha_k.\label{eqlambda}
\end{align}
\item If the loss function $l(x) = x^-$, with $\alpha_k>0$ for all
  $k$, then the value function satisfies $V_t = \widehat V_t$ a.s. for $t=0,\dots,n-1$ where 
\begin{align*}
\widehat V_{n-1} &= \mathbb E[S_n|\mathcal F_{n-1}] - \alpha_n\\
\widehat V_{k-1} &= \max\{\mathbb E[\widehat V_k|\mathcal F_{k-1}],
\mathbb E[\widehat V_k \vee
S_k - \alpha_k|\mathcal F_{k-1}]\}. 
\end{align*}
\end{enumerate}
\end{proposition}

\begin{proof}
\begin{enumerate}
\item The formula for $V_{n-1}$ follows directly from Jensen's
  inequality. We prove the general case by induction. Assume that we
  have already established the equality $V_k=\widehat V_k$. To prove
  the formula for $V_{k-1}$, we first follow Step 1 in the proof of
  Proposition \ref{prop:value function time consistent} to establish
  the existence of $\lambda_k$. Next, let
$$
M = \widehat V_k + \left\{S_k - \widehat V_k + \lambda_k\right\}^+
\mathbf 1_{\mathbb E[l(\widehat V_k -S_k )|\mathcal
  F_{k-1}]>\alpha_k}. 
$$ 
Clearly, $M\geq V_k$. In addition,
\begin{align*}
\mathbb E[l(M-S_k)|\mathcal F_{k-1}] &= \mathbb E[l(\widehat V_k -S_k
+\left\{S_k - \widehat V_k + \lambda_k\right\}^+\mathbf 1_A
)|\mathcal F_{k-1}]\\
& = \mathbf 1_A \mathbb E[l(\lambda\vee(\widehat V_k -S_k)
)|\mathcal F_{k-1}] + \mathbf 1_{A^c} \mathbb E[l(\widehat V_k -S_k
)|\mathcal F_{k-1}]\\
&\leq \alpha_k,
\end{align*}
where we denote $A := \{\mathbb E[l(\widehat V_k -S_k )|\mathcal
  F_{k-1}]>\alpha_k\}\in \mathcal F_{k-1}$. This proves that $V_{k-1}
  \leq \widehat V_{k-1}$. 

The opposite inequality only needs to be shown on the set $A$. As in the proof of
  Proposition \ref{prop:value function time consistent}, we have,
\begin{equation*}V_{k-1}-\mathbb E[\widehat{V}_k|\mathcal F_{k-1}]
=\essinf_{M\in \mathcal F_k, M\geq 0}\{\mathbb E[M|\mathcal F_{k-1}]:
\mathbb E[l(M+\widehat{V}_k-S_k)|\mathcal
F_{k-1}]\leq\alpha_k\}.
\end{equation*}
Let $M$ be any random variable satisfying the constraints. Then, by
convexity,
\begin{align*}
0&\leq \mathbb E[l(\lambda_k \vee (\widehat V_k - S_k)) - l(M+\widehat
V_k - S_k)|\mathcal F_{k-1}] \\
&\leq   \mathbb E[l'(\lambda_k \vee (\widehat V_k - S_k)) (\lambda_k \vee (\widehat V_k - S_k)- M-\widehat
V_k + S_k)|\mathcal F_{k-1}] \\
& = \mathbb E[\mathbf 1_{\lambda_k \geq \widehat V_k -
  S_k}l'(\lambda_k) (\lambda_k -M - \widehat V_k + S_k)-\mathbf 1_{\lambda_k < \widehat V_k -
  S_k}l'(\widehat V_k - S_k) M|\mathcal F_{k-1}]\\
& = l'(\lambda_k)\mathbb E[ (\lambda_k - \widehat V_k +
S_k)^+-M|\mathcal F_{k-1}]
-\mathbb E[(l'(\widehat V_k - S_k)-l'(\lambda_k))^+ M|\mathcal F_{k-1}],
\end{align*}
where $l'(x_0)$ denotes any number belonging to the subdifferential of
the convex function $l$ at the point $x_0$, defined by $$\partial
l(x_0):= \{y:l(x)-l(x_0)\geq y(x-x_0)\ \forall x\}.$$ Since $M\geq 0$, this
implies 
$$
l'(\lambda_k)\mathbb E[ (\lambda_k - \widehat V_k +
S_k)^+-M|\mathcal F_{k-1}] \geq 0.
$$
By the assumption of the proposition, on the set $\{0\in \partial
l(\lambda_k)\}$, $l(\lambda_k)< \alpha_k$ and $\lambda_k$ may not be
the solution of Equation \eqref{eqlambda}. Therefore,
$l'(\lambda_k)<0$ almost surely, and 
$$
\mathbb E[ (\lambda_k - \widehat V_k +
S_k)^+-M|\mathcal F_{k-1}] \leq 0
$$
on $A$, which finishes the proof.
\item The formula for $V_{n-1}$ follows from part 1. When $\mathbb E[(V_k -
S_k)^-|\mathcal F_{k-1}]\leq \alpha_k$, $\mathbb E[V_k|\mathcal
F_{k-1}]\geq \mathbb E[V_k\vee S_k - \alpha_k|\mathcal F_{k-1}]$, and
the formulas hold true. 
Assume then that $\mathbb E[(V_k - S_k)^-|\mathcal F_{k-1}]>\alpha_k$.
First we observe that $\lambda_k<0$ a.s., otherwise on the set where $\lambda_k\geq 0$,
$l(\lambda_k \vee (V_k-S_k)) = 0$ and the conditional expectation
cannot be equal to $\alpha_k$. Therefore, 
\begin{align*}
&\mathbb E\left[\left\{S_k +
    \lambda_k\right\}\vee V_k\Big| \mathcal F_{k-1}\right] = \mathbb
E[S_k|\mathcal F_{k-1}] + \mathbb E[\lambda_k \vee (V_k-S_k)|\mathcal F_{k-1}] \\
&  = E[S_k|\mathcal F_{k-1}] + \mathbb E[(\lambda_k \vee
(V_k-S_k))^+|\mathcal F_{k-1}] - \mathbb E[(\lambda_k \vee
(V_k-S_k))^-|\mathcal F_{k-1}]\\
&  = E[S_k|\mathcal F_{k-1}] + \mathbb E[
(V_k-S_k)^+|\mathcal F_{k-1}] - \mathbb E[l(\lambda_k \vee
(V_k-S_k))|\mathcal F_{k-1}]\\
&= \mathbb E[
V_k \vee S_k|\mathcal F_{k-1}] - \alpha_k.
\end{align*}
\end{enumerate}
\end{proof}

\subsection{Lookback-style constraint}

The lookback constraint involves the maximum value of the loss  function during the whole period. So the dynamic programming structure  takes into account this variable.

Recall that $\mathcal M_{LB}$ denotes the set of all $\mathbb Q$-supermartingales $(M_k)_{k=0}^n$ with 
\[\mathbb E^\mathbb P[\max_{k=1,\dots,n}\{l(M_k - S_k)-\alpha_k\} ] \leq 0.
\]
Define $V_k(N_k,Z_k)$, where $(N_k,Z_k)$ is  a pair of $\mathcal F_k$-measurable random variables,  to be  the essential infimum of all $M_k\in \mathcal F_k$ such that there exist a $\mathbb Q$-supermartingale $(M_t)_{t=k}^n$ and a $\mathbb P$-supermartingale $(N_t)_{t=k}^{n}$ verifying 
\begin{equation}\label{Equ:conditionconditonnelle}\max\big\{Z_k,\max_{t\in\{k+1,\ldots,n\}}\{l(M_t-S_t)-\alpha_t\}\big\}\leq N_n.\end{equation}
We denote the set of all such $\mathbb Q$-supermartingales $(M_t)_{t=k}^n$ by $\mathcal M_{LB,k}(N_k,Z_k)$.
By convention, let
\[V_n(N_n,Z_n)=(+\infty)\indic_{\{Z_n>N_n\}}+(-\infty)\indic_{\{Z_n\leq N_n\}}.\]

\begin{proposition} 
\[V_0(0,-\infty)=\inf_{(M_k)_{k=0}^n\in\mathcal M_{LB}}M_0.\]
\end{proposition}
\begin{proof}
Let $(M_k)_{k=0}^n$ and $(N_k)_{k=0}^n$ be respectively $\mathbb Q$ and $\mathbb P$ supermartingales which verify \eqref{Equ:conditionconditonnelle} with $(N_0,Z_0)=(0,-\infty)$. Then
\[\max_{k\in\{1,\ldots,n\}}\{l(M_k-S_k)-\alpha_k\}\leq N_n.\]
Therefore
\[0=N_0\geq\mathbb E^{\mathbb P}[N_n]\geq\mathbb E^{\mathbb P}[\max_{k\in\{1,\ldots,n\}}\{l(M_k-S_k)-\alpha_k\}].\]
Hence $(M_k)_{k=0}^n\in\mathcal M_{LB}$. Therefore  \[V_0(0,-\infty)\geq\inf_{(M_k)_{k=0}^n\in\mathcal M_{LB}}M_0.\]

Conversely, given a $\mathbb Q$-supermartingale $(M_k)_{k=0}^n$ in $\mathcal M_{LB}$, we can choose
\[N_n=\max_{k\in\{1,\ldots,n\}}\{l(M_k-S_k)-\alpha_k\}\]
and $N_k=\mathbb E^{\mathbb P}[N_n|\mathcal F_k]$, $k\in\{1,\ldots,n-1\}$ and $N_0=0$ to construct a $\mathbb P$-supermartingale $(N_k)_{k=0}^n$ verifying the condition \eqref{Equ:conditionconditonnelle} with $(N_0,Z_0)=(0,-\infty)$. Thus  $M_0\geq V_0(0,-\infty)$. The result is proved.
\end{proof}

\begin{proposition}\label{lookback: recursive}
$V_k(N_k,Z_k)$ equals the essential infimum of all $\mathbb E^{\mathbb Q}[M|\mathcal F_k]$, $M\in\mathcal F_{k+1}$ such that there exists a random variable $N\in\mathcal F_{k+1}$ satisfying\begin{equation}\label{Equ:conditionsLB}
\begin{cases}
\mathbb E^{\mathbb P}[N|\mathcal F_k]=N_k,\\
M\geq V_{k+1}\big(N,\max(Z_k,l(M-S_{k+1})-\alpha_{k+1})\big).
\end{cases}
\end{equation}
\end{proposition}
\begin{proof} Denote by $\widetilde V_k(N_k,Z_k)$ the essential infimum defined in the theorem.
Let $M$ and $N$ be random variables which verify  \eqref{Equ:conditionsLB}. By an argument similar to Lemma \ref{Lem:filtrant}, 
there exists a decreasing sequence of random variables $(M_{k+1}^{(m)})_{m\in\mathbb N}$ in $\mathcal M_{LB,k+1}(N,\max(Z_k,l(M-S_{k+1})-\alpha_{k+1}))$ such that $M\geq\inf_{m\in\mathbb N}M_{k+1}^{(m)}$. Moreover, there exist $\mathbb Q$-supermartingales $(M_t^{(m)})_{t=k+1}^n$ and $\mathbb P$-supermartingales $(N_t^{(m)})_{t=k+1}^n$ verifying 
\begin{equation}\label{Equ:conditionconditonnelle2}\max\big\{Z_k,l(M-S_{k+1})-\alpha_{k+1},\max_{t\in\{k+2,\ldots,n\}}(l(M_t^{(m)}-S_t)-\alpha_t)\big\}\leq N_n^{(m)}\end{equation}
and $N_{k+1}^{(m)}=N$. Let $M_k^{(m)}=\mathbb E^{\mathbb Q}[M_{k+1}^{(m)}|\mathcal F_k]$ and $N_k^{(m)}=N_k$. The processes $(M_t^{(m)})_{t=k}^n$ and  $(N_t^{(m)})_{t=k}^{n}$ are respectively $\mathbb Q$ and $\mathbb P$ supermartingales which satisfy the condition \eqref{Equ:conditionconditonnelle}. Therefore,  $M_k^{(m)}\geq V_k(N_k,Z_k)$ for any $m\in\mathbb N$, which implies $\mathbb E^{\mathbb Q}[M|\mathcal F_k]\geq V_k(N_k,Z_k)$. So $\widetilde V_k(N_k,Z_k)\geq V_k(N_k,Z_k)$.

Conversely, let $(M_t)_{t=k}^n$ and $(N_t)_{t=k}^n$ be respectively $\mathbb Q$-supermartingale and $\mathbb P$-supermartingale which verify the condition \eqref{Equ:conditionconditonnelle}. Without loss of generality, we may assume  $M_k=\mathbb E^{\mathbb Q}[M_{k+1}|\mathcal F_k]$ and that $(N_t)_{t=k}^n$ is a $\mathbb P$-martingale
. Note that 
$\max\big\{Z_k,\max_{t\in\{k+1,\ldots,n\}}(l(M_t-S_t)-\alpha_t)\big\}$
can also be written as
\[\max\big\{\max(Z_k,l(M_{k+1}-S_{k+1})-\alpha_{k+1}),\max_{t\in\{k+2,\ldots,n\}}(l(M_t-S_t)-\alpha_t)\big\}.\]
By the definition of $V_{k+1}$, 
\[M_{k+1}\geq V_{k+1}(N_{k+1},\max(Z_k,l(M_{k+1}-S_{k+1})-\alpha_{k+1})),\]
which shows that $M_{k+1}$ and $N_{k+1}$ verify the conditions in \eqref{Equ:conditionsLB}. Therefore $M_k=\mathbb E^{\mathbb Q}[M_{k+1}|\mathcal F_k]\geq\widetilde V_k(N_k,Z_k)$. Since $(M_t)_{t=k}^n$ is arbitrary, we obtain $V_{k}(N_k,Z_k)\geq\widetilde V_k(N_k,Z_k)$.

\end{proof}

\section{Explicit examples for $n=2$}
\label{examples2d.sec}

In this section, we apply the dynamic programming results obtained in
the previous section under the three types of constraints  and we give
the explicit form of the smallest hedging portfolio value in the setting of two time steps.  

\subsection{European-style constraint} 

In the following, we assume that the function $l$ satisfies Assumption \ref{assumption 2}. We have by Corollary \ref{Cor:dynamic},
$$
V_1(\alpha_2)=\essinf_{M\in\mathcal F_2}\{\mathbb E^{\mathbb
  Q}[M|\mathcal F_1] \text{ s.t. } \mathbb E^{\mathbb
  P}[l(M-S_2)|\mathcal F_1]\leq \alpha_2\}
$$
By Proposition \ref{prop:value function time consistent}, 
\begin{equation}\label{V1,n=2}V_1(\alpha_2)=\mathbb E^{\mathbb
    Q}[S_2|\mathcal F_1]+\mathbb E^{\mathbb
    Q}[I(c(\alpha_2)Z_2/Z_1)|\mathcal F_1],\end{equation}
where 
 $c(\alpha_2)$ is the $\mathcal F_1$-measurable random variable such that  
\[\mathbb E^{\mathbb P}[l(I(c(\alpha_2)Z_2/Z_1))|\mathcal F_1]=\alpha_2.\]
and 
$$
Z_2 = \frac{d\mathbb Q}{d\mathbb P}; \quad Z_1 = \mathbb E^{\mathbb P}
[\frac{d\mathbb Q}{d\mathbb P}|\mathcal F_1].
$$

We then use  \eqref{Equ:europfinale} again to compute the value of $V_0(\alpha_1,\alpha_2)$, which identifies with the essential infimum of $\mathbb E^{\mathbb Q}[M]$, $M\in\mathcal F_1$ such that $\mathbb E[l(M-S_1)]\leq \alpha_1$ and that there exists $N\in\mathcal F_1$ verifying $\mathbb E[N]=\alpha_2$ and $M\geq V_1(N)$. Namely, 
\begin{equation}\label{V_0EU, n=2}V_0(\alpha_1,\alpha_2)=\inf_{M\in\mathcal F_1}\{\mathbb E^{\mathbb Q}[M]\,:\,
\mathbb E^{\mathbb P}[l(M-S_1)]\leq \alpha_1,\,\mathbb E^{\mathbb P}[V_1^{-1}(M)]\leq \alpha_2\}.
\end{equation}

\subsubsection{Exponential utility case}

Consider the loss function where $l(x)=e^{-px}-1$, with $p>0$. Direct computations yield $l'(x)=-pe^{-px}$, $I(x)=-\frac{1}{p}\log(-x/p)$ and $l(I(x))=-x/p-1$. Therefore, 
\[c(\alpha_2)=-{p(1+\alpha_2)},\]
from which we deduce
\begin{equation}\label{Equ:v1exp}V_1(\alpha_2)=\mathbb E^{\mathbb Q}[S_2|\mathcal F_1]-\frac{\log(\alpha_2+1)}{p}-\frac 1p{\mathbb
  E^{\mathbb P}[Z_2/Z_1\log(Z_2/Z_1)|\mathcal F_1]}.\end{equation}
Denote by 
\[D_1=-\frac 1p{\mathbb E^{\mathbb P}[Z_2/Z_1\log(Z_2/Z_1)|\mathcal F_1]}.\]
From \eqref{Equ:v1exp} we have that
\[V_1^{-1}(x)=l(x-\mathbb E^{\mathbb Q}[S_2|\mathcal F_1]-D_1).\]
Hence $V_0(\alpha_1,\alpha_2)$ is the infimum of the expectations
$\mathbb E^{\mathbb Q}[M]=\mathbb E^{\mathbb P}[Z_1M]$, where $M\in\mathcal F_1$ such that
\begin{equation}\label{Equ:constraint}\mathbb E^{\mathbb P}[l(M-S_1)]\leq \alpha_1\;\text{ and }\;
\mathbb E[l(M-\mathbb E^{\mathbb Q}[S_2|\mathcal F_1]-D_1)]\leq \alpha_2.\end{equation}
We will use Lagrange multiplier method (with Kuhn-Tucker conditions) to study this problem. If $M\in\mathcal F_1$ realizes the infimum of $\mathbb E^{\mathbb Q}[M]$ subjected to the constraints \eqref{Equ:constraint}, 
then $M$ verifies the first order condition
\[Z_1-p\lambda e^{-p(M-S_1)}-p\mu e^{-p(M-\mathbb E^{\mathbb Q}[S_2|\mathcal F_1]-D_1)}=0,\]
where $\lambda$ and $\mu$ are respectively positive solutions to the equations
\[\mathbb E^{\mathbb P}\bigg[\frac{Z_1}{
p(\lambda+\mu \exp(p(\mathbb E^{\mathbb Q}[S_2|\mathcal F_1]-S_1+D_1)))}\bigg]=1+\alpha_1\]
and
\[\mathbb E^{\mathbb P}\bigg[\frac{Z_1}{p(\lambda\exp(S_1-\mathbb E^{\mathbb Q}[S_2|\mathcal F_1]+D_1)+\mu)}\bigg]=1+\alpha_2.\] 
We introduce the notation $X=\exp(p(\mathbb E^{\mathbb Q}[S_2|\mathcal F_1]-S_1+D_1))$ and rewrite the above equations as
\[\mathbb E^{\mathbb Q}\Big[\frac{1}{\lambda+\mu X}\Big]=p(1+\alpha_1)\;\quad\text{and}\quad\;\mathbb E^{\mathbb Q}\Big[\frac{X}{\lambda+\mu X}\Big]=p(1+\alpha_2),\]
from which we obtain
\begin{equation}\label{Equ:equlinear}\lambda p(1+\alpha_1)+\mu p(1+\alpha_2)=1\end{equation}
by taking a linear combination. Moreover, the quotient of the two equations shows that $\lambda/\mu$ verifies the following equation (provided that $\mu\neq 0$)
\[\frac{f(t)}{1-tf(t)}=\frac{1+\alpha_1}{1+\alpha_2},\]
where $f(t)=\mathbb E^{\mathbb Q}[(t+X)^{-1}]$. Note that
\[\lim_{t\rightarrow 0_+}\frac{f(t)}{1-tf(t)}=\mathbb E^{\mathbb Q}[X^{-1}],\qquad
\lim_{t\rightarrow+\infty}\frac{f(t)}{1-tf(t)}=\mathbb E^{\mathbb Q}[X]^{-1}.\]

\begin{itemize}
\item Suppose that 
\[\mathbb E^{\mathbb Q}[X^{-1}]>\frac{1+\alpha_1}{1+\alpha_2}>\mathbb E^{\mathbb Q}[X]^{-1}.\]
The two constraints in \eqref{Equ:constraint} are saturated. We can use numerical methods to compute the value of $\lambda/\mu$ and then use the relation \eqref{Equ:equlinear} to find explict value of $\lambda$ and $\mu$.
\item If 
\[\frac{1+\alpha_1}{1+\alpha_2}\leq \mathbb E^{\mathbb Q}[X]^{-1},\]
then only the first constraint in \eqref{Equ:constraint} is saturated and one should take $\mu=0$. In this case the first order condition becomes
\[Z_1-p\lambda e^{-p(M-S_1)}=0,\]
with $\lambda=(p(1+\alpha_1))^{-1}$. Therefore, one has
\[M=S_1-\frac{1}{p}\log((1+\alpha_1)Z_1).\]
Thus we obtain
\[V_0(\alpha_1,\alpha_2)=\mathbb E^{\mathbb Q}[S_1]-\frac{1}{p}\log(1+\alpha_1)-\frac{1}{p}
\mathbb E^{\mathbb P}[Z_1\log Z_1].\]
\item If
\[\frac{1+\alpha_1}{1+\alpha_2}\geq\mathbb E^{\mathbb Q}[X^{-1}],\]
then only the second constraint in \eqref{Equ:constraint} is saturated, and one has $\lambda=0$. In this case the first order condition becomes
\[Z_1-p\mu e^{-p(M-\mathbb E^{\mathbb Q}[S_2|\mathcal F_1]-D_1)}=0\]
with $\mu=(p(1+\alpha_2))^{-1}$. Therefore, one obtains
\[\begin{split}M&=\mathbb E^{\mathbb Q}[S_2|\mathcal F_1]+D_1-\frac{1}{p}\log((1+\alpha_1)Z_1)\\
&=\mathbb E^{\mathbb Q}[S_2|\mathcal F_1]-\frac{1}{p}\mathbb
E^{\mathbb P}[Z_2/Z_1\log (Z_2/Z_1)|\mathcal F_1]-\frac{1}{p}\log((1+\alpha_1)Z_1).
\end{split}\]
Finally, we obtain
\[V_0(\alpha_1,\alpha_2)=\mathbb E^{\mathbb
  Q}[S_2]-\frac{1}{p}\log(1+\alpha_2)-\frac 1p\mathbb E^{\mathbb
  P}[Z_2\log Z_2].\]
\end{itemize}

\subsection{Time-consistent constraint}
The computation of $V_1$ is identical to \eqref{V1,n=2} in the European case: $V_1 =
V_1(\alpha_2)$. By Proposition \ref{Pro:TCprogdy},
{\begin{equation}\label{V_0TC,n=2}
V_0 = \inf_{M \in \mathcal F_1}\{\mathbb E^{\mathbb Q}[M] : \mathbb
E^{\mathbb P}[l(M-S_1)]\leq \alpha_1, V_1^{-1}(M) \leq \alpha_2\}.
\end{equation}}
\subsubsection{Exponential utility}
From \eqref{Equ:v1exp} as in the European case,
$$
V_1 = \mathbb E^{\mathbb Q}[S_2|\mathcal F_1]-\frac{\log(\alpha_2+1)}{p}
+D_1.
$$
By Proposition \ref{prop:value function time consistent}, $V_0$ may be found using the
following algorithm:
\begin{itemize}
\item Compute $\mathbb E^{\mathbb P}[l(V_1 - S_1)]$. 
\item If $\mathbb
  E^{\mathbb P}[l(V_1 - S_1)] \leq \alpha_1$ then the constraint at
  date $1$ is not binding, and the solution is given by $V_0=\mathbb
  E^{\mathbb Q}[V_1]$. 
\item If $\mathbb
  E^{\mathbb P}[l(V_1 - S_1)] > \alpha_1$ then the constraint at
  date $1$ is binding and we proceed as follows:
\begin{itemize}
\item Compute $\lambda \in \mathbf R$ by solving the equation
$$
\mathbb E^{\mathbb P} [l\left(I(\lambda Z_1)\vee (V_1 - S_1)\right)] =
\alpha_1.
$$
\item The solution is given by 
$$
V_0 = \mathbb
  E^{\mathbb Q}[V_1] + \mathbb E^{\mathbb Q}[\{S_1 - V_1 + I(\lambda
  Z_1)\}^+]. 
$$

\end{itemize}
\end{itemize}

\subsection{Lookback-style constraint}
By Proposition \ref{lookback: recursive},
$$
V_{1}(y,z) = \inf_{M\in \mathcal F_2}\left\{\mathbb E^{\mathbb Q}
  [M|\mathcal F_1] : \mathbb E^{\mathbb P}[\max(z, l(M-S_2) -
  \alpha_2)|\mathcal F_1]< y\right\}
$$

So the value function  $V_1(0,-\infty)$ coincides with \eqref{V1,n=2} as in the European and the American cases.  
We then have the minimal capital at the initial time as
$$
V_0(0,-\infty) = \inf_{M\in \mathcal
F_{1}}\{\mathbb E^{\mathbb Q}[M] :\exists N \in \mathcal F_1 : \mathbb
E^{\mathbb P}[N]=0, M\geq V_1(N,l(M-S_1)-\alpha_1)\}. 
$$
However, it is more difficult to obtain explicit results as \eqref{V_0EU, n=2} ou \eqref{V_0TC,n=2} for the lookback-style constraint. 
\subsection{Risk-neutral case}

As we have already mentioned, in the three cases, the value functions $V_1$ at time $t=1$ coincide. In the risk-neutral case where $\mathbb P=\mathbb Q$, we have $Z_1=Z_2=1$, so 
\[V_1(\alpha_2)=\mathbb E[S_2|\mathcal F_1]+l^{-1}(\alpha_2).\]
For the initial value at time $t=0$, we have the following results : 
\begin{itemize}
\item {\it European-style constraint : }
\[V_0^{EU}=\inf_{M\in\mathcal F_1}\{\mathbb E[M] : \mathbb E[l(M-S_1)]\leq \alpha_1, \mathbb E[l(M-\mathbb E[S_2|\mathcal F_1])]\leq \alpha_2\}\]

\item {\it American-style constraint : }
\[V_0^{TC}=\inf_{M\in\mathcal F_1}\{\mathbb E[M] : \mathbb E[l(M-S_1)]\leq \alpha_1, l(M-\mathbb E[S_2|\mathcal F_1])\leq \alpha_2\}\]

\item {\it Lookback-style constraint : }
\[V_0^{LB}= \inf_{M\in \mathcal
F_{1}}\{\mathbb E [M]: \mathbb E[\max\{ l(M-S_1)-\alpha_1, l(M-\mathbb E[S_2|\mathcal
F_1])-\alpha_2\}]\leq 0\}.\]
\end{itemize}

In the risk-neutral case, the results hold under Assumption \ref{assumption 1} instead of Assumption \ref{assumption 2}. We are particularly interested in the loss function 
$l(x)=(-x)_+$ where explicit results are obtained, thanks to Lemma
\ref{2dim} in the appendix (for the time-consistent constraint the
result follows by taking $\beta = 0$ and $Y = \mathbb E[S_2|\mathcal
F_1] - \alpha_2$) : 
\begin{align*}
V_0^{EU}&=\max\big(\mathbb E[S_1]-\alpha_1,\mathbb E[S_2]-\alpha_2,\mathbb E[(S_1\vee\mathbb E[S_2|\mathcal F_1])]-\alpha_1-\alpha_2\big),\\
V_0^{TC} &= \max\big(\mathbb E[S_2 ]- \alpha_2,
\mathbb E[S_1 \vee (\mathbb E[S_2|\mathcal F_1] -
\alpha_2)]-\alpha_1\big),\\
V_0^{LB} &= \mathbb E[\max\big(S_1 - \alpha_1,(\mathbb
E[S_2|\mathcal F_1] - \alpha_2)\big)]. 
\end{align*}

\subsection{Numerical illustration}

We compare the cost for hedging two
objectives under the three probabilistic constraints in a numerical example. 

In Figure \ref{exploss.fig}, we consider the risk-neutral case where $\mathbb P = \mathbb Q$  and the loss function  is given by $l(x) = (-x)_+$. The model is as follows : $S_1 =
S_0 e^{\sigma Z_1 - \frac{\sigma^2}{2}}$ and $S_2 =
S_0 e^{\sigma Z_2 - \frac{\sigma^2}{2}}$ where $S_0 = 100$, $\sigma =
0.2$ and $Z_1$ and $Z_2$ are standard normal random variables with
correlation $\rho = 50\%$. 
We fix the loss tolerance at the first date
to be $\alpha_1=  5$ and plot the 
hedging value $V_0$ for three different constraint styles as a function of the loss
tolerance at the second date $\alpha_2$. 

Not surprisingly, all the three curves are decreasing w.r.t. the constraint level $\alpha_2$.  Moreover, among the three constraints we consider, the European constraint corresponds to the lowest hedging cost and the lookback constraint to the highest one, which is coherent with Proposition \ref{pro:inclusion}. Finally, the cost of almost sure hedging is higher than the three probabilistic constraints. 


\begin{figure}
\centerline{\includegraphics[width=0.8\textwidth]{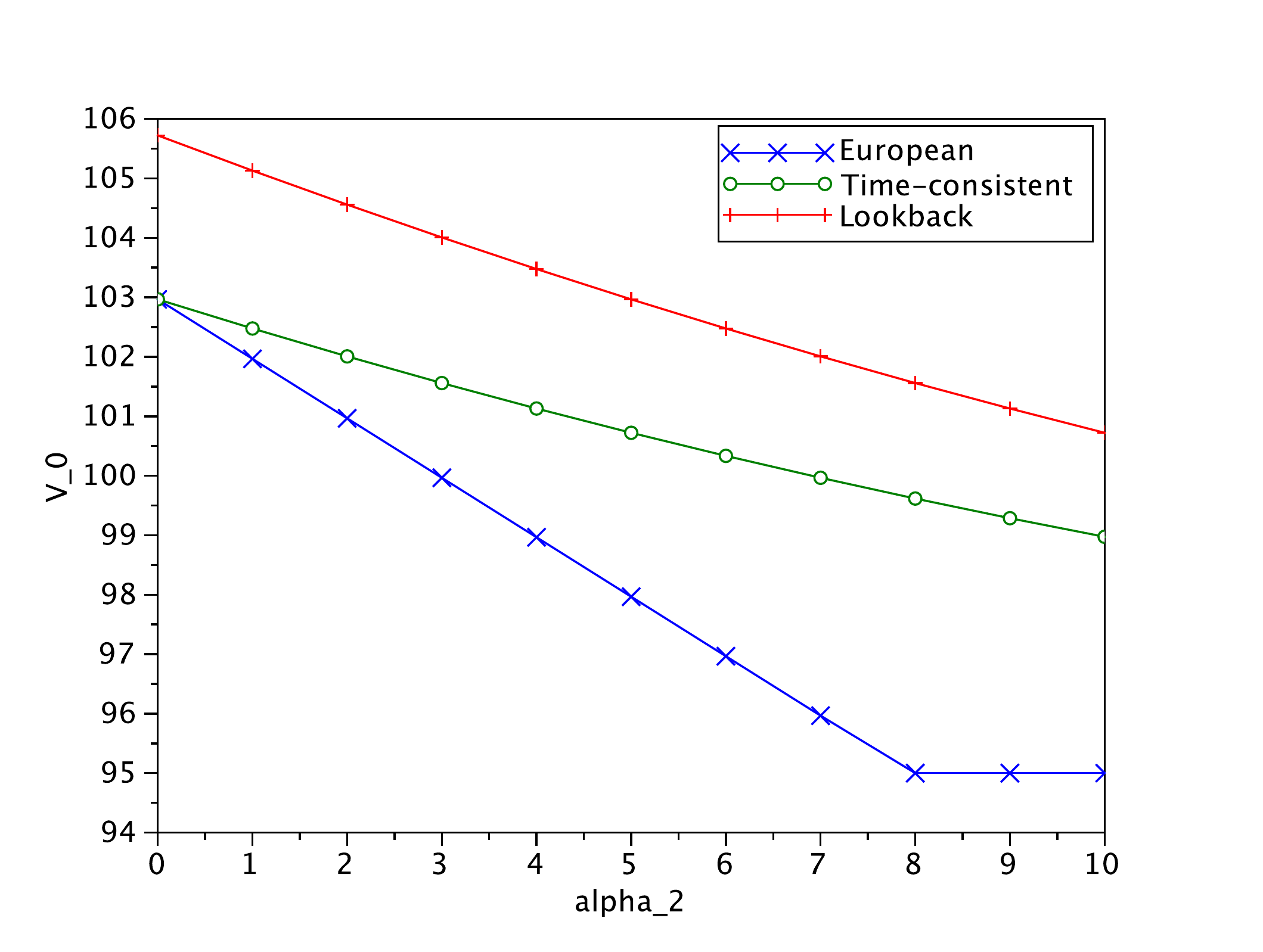}}
\caption{Cost of hedging two objectives with probabilistic constraints
of three different styles. The cost of hedging both constraints in an
almost sure way is equal to $107.966$ in this case.}
\label{exploss.fig}
\end{figure}

\section{Explicit examples for arbitrary $n$}
\label{examplesnd.sec}
For the multi-objective hedging problem, it is in general difficult to
get explicit solutions by using the recursive formulas. In this
section, we first present a situation where the explicit solution may
be obtained for the expected loss and then extend our framework to
discuss the conditional Value at Risk. 

Our first result deals with the expected loss constraint  given as a call function under a risk-neutral probability.

\begin{proposition}\label{incloss.cor}Let $l(x)=(-x)^+$.  
Assume that $\mathbb P=\mathbb Q$,   $\alpha_1,\dots,\alpha_n \geq 0$
and that the process $(S_k)_{k=0}^n$ is non-decreasing. Then,
\begin{align}
V_0^{EU} = \max_{k\in\{1,\dots,n\}} \{\mathbb E[S_k] - \alpha_k\}\label{exploss}
\end{align}

\end{proposition}

\begin{proof}
For any $k\in\{1,\cdots, n\}$, we have
\begin{align*}
V_0^{EU}&=\inf_{(M_k)_{k=0}^n\in\mathcal M_{EU}}\{ M_0: \mathbb E[(S_k-M_k)^+]\leq\alpha_k,\,k=1,\dots,n\} \\&
\geq \inf_{M\in \mathcal F_k}\{ \mathbb E[M]: \mathbb E[(S_k-M)^+]\leq
\alpha_k\} = \mathbb E[S_k] - \alpha_k,
\end{align*}
and therefore, 
$$
V_0^{EU} \geq \max_{k\in\{1,\dots,n\}} \{\mathbb E[S_k] - \alpha_k\}
$$
On the other hand, Jensen's inequality yields
$\mathbb E[(S_k-M_k)^+] \leq \mathbb E[(S_k-M_n)^+]$, so 
\begin{align*}
V_0^{EU}&\leq \inf_{(M_k)_{k=0}^n\in\mathcal M_{EU}}\{ M_0: \mathbb E[(S_k-M_n)^+]\leq
\alpha_k,\,k=1,\dots,n\} \\
&=\inf_{M\in \mathcal F_n}\{ \mathbb E[M]: \mathbb E[(S_k-M)^+]\leq
\alpha_k,\,k=1,\dots,n\}.
\end{align*}
Since $(S_t)_{0\leq t\leq n}$ is non-decreasing, then  \eqref{exploss} follows by Lemma \ref{exploss.lm} in Appendix.
\end{proof}

We have up to now considered risk constraints given by (conditional)
expected loss. In the literature and in practice, the risk constraint
is often described by using risk measures. Let us now extend our
framework to cover loss constraints expressed in terms of the
conditional Value at Risk. We once again assume that $\mathbb P =
\mathbb Q$ and denote by $\mathcal M_{\mathrm{CVaR}}$ the set of all supermartingales $(M_k)_{k=0}^n$ such that $\mathrm{CVaR}_\lambda[(S_k-M_k)^+]\leq\alpha_k$ for any $k\in\{1,\ldots,n\}$, where $\text{CVaR}_{\lambda}$ is the risk measure defined by
\begin{align}
\text{CVaR}_{\lambda} (X)= \sup_{\tilde{\mathbb P}} \{\mathbb E^{\tilde{\mathbb P}} [X]: \frac{d \tilde{\mathbb P}}{d\mathbb P}
\leq \frac{1}{1-\lambda}\},\label{cvarproba}
\end{align}
where $X \in L^1$ represents the loss.
It is also well known to  satisfy
$$
\text{CVaR}_{\lambda} (X) = \text{VaR}_{\lambda} (X) +
\frac{1}{1-\lambda} \mathbb E[(X-\text{VaR}_{\lambda} (X) )^+]
$$
and
\begin{align*}
\text{CVaR}_{\lambda} (X) = \inf_{z\in \mathbf R}
\left\{\frac{1}{1-\lambda}\mathbb E[(X-z)^+] + z\right\}.
\end{align*}
The minimum above is attained, so that one can
write 
\begin{align}
\text{CVaR}_{\lambda} (X) = \min_{z\in \mathbf R}
\left\{\frac{1}{1-\lambda}\mathbb E[(X-z)^+] + z\right\}.\label{varexp}
\end{align}
Moreover, the simplified formula
$$
\text{CVaR}_{\lambda} (X) = \mathbb E[X | X > \text{VaR}_{\lambda} (X)]
$$
holds whenever $X$ has continuous distribution.
 
We define
\begin{align}
\label{cvar}
V_0^{\text{CVaR}}:=\inf_{(M_k)_{k=0}^n\in\mathcal M_{\mathrm{CVaR}}}M_0
\end{align}

Let us start with some simple observations:
\begin{itemize}
\item From \eqref{cvarproba} one deduces immediately that 
  $\text{CVaR}_{\lambda} (X) \geq \mathbb E[X]$, which means (see the
  proof of Proposition
  \ref{incloss.cor}) that
\begin{align}
V_0^{\text{CVaR}} \geq \max_{k\in\{1,\dots,n\}} \{\mathbb E[S_k] - \alpha_k\}.\label{simpleobs}
\end{align}
\item For $n=1$, the choice $M^*_1 = S_1 - \alpha_1$ attains the
  optimal value in \eqref{exploss}. But in this case
  $\text{CVaR}_\lambda [(S_1-M_1)^+] = \alpha_1$, which implies that
  the lower bound above is also attained:
\begin{align*}
V_0^{\text{CVaR}} =  \mathbb E[S_1] - \alpha_1.
\end{align*}
\item Since in our case the loss is always positive, the equality
  \eqref{varexp} may be simplified as
\begin{align*}
\text{CVaR}_{\lambda} ((S-M)^+) =
\min_{z\geq 0}
\left\{\frac{1}{1-\lambda}\mathbb E[(S-M-z)^+] + z\right\}.
\end{align*}
\end{itemize} 

The following representation reduces the problem \eqref{cvar} to that
of hedging under expected loss constraints. 
\begin{proposition}\label{Glemma}
\begin{align*}
V_0^{\mathrm{CVaR}} =
\inf_{0\leq z \leq \alpha} G(z_1,\dots,z_n),
\end{align*}
where $0\leq z \leq \alpha$ is a shorthand for $z_k\in [0,\alpha_k]$, $k\in\{1\ldots n\}$ and
\begin{multline*}
G(z_1,\dots,z_n) \\ = \inf_{(M_k)_{k=0}^n \, \text{supermartingale}}\{M_0 : \frac{1}{1-\lambda}\mathbb E[(S_k - M_k - z_k)^+] + z_k \leq \alpha_k, \,k=1\ldots n\}
\end{multline*}
 
\end{proposition}
\begin{proof}
For every $z_1,\dots,z_n$ fixed, the set $\mathcal M_{\mathrm{CVaR}}$ contains the set of supermartingales $(M_k)_{k=0}^n$ satisfying 
$\frac{1}{1-\lambda}
\mathbb E[(S_k - M_k - z_k)^+] + z_k \leq \alpha_k,\, k\in\{1\ldots n\}$,
which shows that
\begin{align*}
V_0^{\mathrm{CVaR}}\leq
\inf_{0\leq \bar z \leq \bar\alpha} G(z_1,\dots,z_n). 
\end{align*}
On the other hand, for any
$\varepsilon >0$, there exists an $\mathbb F$-supermartingale $(M_k)_{k=0}^n$ satisfying 
$\mathrm{CVaR}_\lambda [(S_k-M_k)^+]\leq
\alpha_k,\, k\in \{1,\dots,n\}$ and such that 
$$
M_0 < \varepsilon + V_0^{\mathrm{CVaR}} 
$$
Since the infimum in \eqref{exploss} is attained, we can also find
$z_1,\dots,z_n$ such that 
$\frac{1}{1-\lambda}
\mathbb E[(S_k - M_k - z_k)^+] + z_k \leq \alpha_k,\, k\in\{1\ldots n\}$. But
this means that 
\begin{align*}
V_0^{\mathrm{CVaR}} + \varepsilon >
\inf_{0\leq \bar z \leq \bar\alpha} G(z_1,\dots,z_n),
\end{align*}
which completes the proof since $\varepsilon$ is arbitrary. 
\end{proof}

We close this section by presenting a particular case when the
solution of problem \eqref{cvar} is completely explicit. 
\begin{proposition}
If $(S_k-\alpha_k)_{k=1}^n$ is an increasing
process, then
$$
V_0^{\mathrm{CVaR}} = \mathbb E[S_n] - \alpha_n. 
$$
\end{proposition}
\begin{proof}
If the process
$(S_k-z_k)_{k=1}^n$ is increasing, by Proposition \ref{incloss.cor}, 
$$
G(z_1,\dots,z_n) = \max_{k\in\{1,\cdots,n\}} \{\mathbb E[S_k] - (1-\lambda) \alpha_k -
\lambda z_k\} = \mathbb E[S_n] - (1-\lambda)\alpha_n - \lambda z_n.
$$
Observe now that 
$$
G(\alpha_1,\dots,\alpha_n) = \mathbb E[S_n] - \alpha_n.
$$
Combining this with \eqref{simpleobs}, the proof is complete.  
\end{proof}

\appendix

\section*{Appendix}

\begin{lemma}\label{2dim}
Let $(\Omega,\mathcal F, \mathbb P)$ be a probability space, $X,Y \in \mathcal F$ and $\alpha,\beta \geq
0$. Then,
\begin{multline*}
V_0 = \inf_{M\in \mathcal F}\{ \mathbb E[M]: \mathbb E[(X-M)^+]\leq
\alpha,\mathbb E[(Y-M)^+]\leq \beta \} \\= \max\{\mathbb E[X-\alpha],
\mathbb E[Y-\beta], \mathbb E[X\vee Y - \alpha - \beta]\}. 
\end{multline*}
and
\begin{multline*}
V'_0 = \inf_{M\in \mathcal F}\{ \mathbb E[M]: \mathbb E[\max((X-M)^+-
\alpha,(Y-M)^+- \beta)]\leq 0 \} \\= \mathbb E[(X-\alpha)\vee (Y - \beta)]. 
\end{multline*}
\end{lemma}
\begin{proof}
\emph{First part.}\quad 
By considering only one of the two constraints, we have clearly
$$
V_0 \geq \max\{\mathbb E[X-\alpha],
\mathbb E[Y-\beta]\}.
$$
In addition,
$$
(X-M)^+ + (Y-M)^+ \geq X\vee Y - M
$$
and so
$$
V_0 \geq \mathbb E[X\vee Y - \alpha - \beta].
$$
To finish the proof, we need to find $M^*$ which satisfies the
constraint and realizes the equality. 

Assume first that 
$$\mathbb E[X\vee Y - \alpha -
\beta] \geq \mathbb E[X-\alpha]\quad \text{et}\quad \mathbb E[X\vee Y - \alpha -
\beta] \geq \mathbb E[Y-\beta].$$ 
This is equivalent to  
$$\mathbb
E[(X-Y)^+]\geq \alpha\quad \text{and}\quad \mathbb E[(Y-X)^+]\geq
\beta.$$ 
In this case we can take
$$
M^* = X \vee Y - \alpha\frac{(X-Y)^+}{\mathbb E[(X-Y)^+]} - \beta
\frac{(Y-X)^+}{\mathbb E[(Y-X)^+]}
$$

If 
$$\mathbb E[X-\alpha]\geq \mathbb E[X\vee Y - \alpha -
\beta]  \quad \text{and}\quad \mathbb \mathbb E[X-\alpha] \geq \mathbb E[Y-\beta],
$$
we take 
$$
M^* = X\wedge Y-\alpha+ \mathbb E[(X-Y)^+].
$$ 
If finally
$$\mathbb E[Y-\beta]\geq \mathbb E[X\vee Y - \alpha -
\beta]  \quad \text{and}\quad \mathbb \mathbb E[Y-\beta] \geq \mathbb E[X-\alpha],
$$
then we can take 
$$
M^* = X\wedge Y-\beta+ \mathbb E[(Y-X)^+].
$$ 

\noindent \emph{Second part.}\quad Since $M = (X-\alpha)\vee
(Y-\beta)$ satisfies the constraint, $V'_0 \leq \mathbb E[(X-\alpha)\vee
(Y-\beta)]$. On the other hand,
\begin{align*}
V'_0 &= \inf_{M\in \mathcal F}\{ \mathbb E[M]: \mathbb E[\max((X-M)^+-
\alpha,(Y-M)^+- \beta)]\leq 0 \} \\ &\geq \inf_{M\in \mathcal F}\{ \mathbb E[M]: \mathbb E[\max(X-M-
\alpha,Y-M- \beta)]\leq 0 \} \\
&= \inf_{M\in \mathcal F}\{ \mathbb E[M]: \mathbb E[\max(X-
\alpha,Y- \beta)]\leq M \}\\
& = \mathbb E[(X-\alpha)\vee
(Y-\beta)].
\end{align*}
\end{proof}

We now consider the case when $n$ is arbitrary but the objectives are ordered.

\begin{lemma}\label{exploss.lm}
Let $(\Omega,\mathcal F, \mathbb P)$ be a probability space, $Z_1,\dots,Z_n \in \mathcal F$ with $Z_1\leq Z_2\leq\dots\leq Z_n
$ a.s. and $\alpha_1,\dots,\alpha_n \geq
0$. Then,
$$
\inf_{M\in \mathcal F}\{ \mathbb E[M]: \mathbb E[(Z_k-M)^+]\leq
\alpha_k,k=1,\dots,n\} = \max_{k=1,\dots,n} \{\mathbb E[Z_k] - \alpha_k\}
$$
\end{lemma}
\begin{proof}${}$
\begin{enumerate}
\item Assume that there exists $k\in\{1,\cdots, n\}$ such that  
\begin{align}
\mathbb E[Z_k] - \alpha_k \geq  \mathbb E[Z_{k+1}] - \alpha_{k+1}\label{wrong}
\end{align}
Then the constraint $\mathbb E[(Z_{k} - M)^+] \leq \alpha_{k}$ implies
$\mathbb E[(Z_{k+1} - M)^+] \leq \alpha_{k+1}$ since
$$
\mathbb E[(Z_{k+1} - M)^+] \leq \mathbb E[(Z_{k} - M)^+] + \mathbb
E[Z_{k+1}-Z_k] \leq \alpha_{k} + \alpha_{k+1} - \alpha_{k} = \alpha_{k+1}.
$$
We can then remove the constraint $\mathbb E[(Z_{k+1} - M)^+]
\leq \alpha_{k+1}$ without modifying the value function. Repeating the
same argument for all other indices $k$ satisfying \eqref{wrong},
we can assume with no loss of generality that,
\begin{align}
\mathbb E[Z_1] - \alpha_1 < \dots< \mathbb E[Z_n] - \alpha_n.\label{right}
\end{align}
We then need to prove that
$$
\inf_{M\in \mathcal F}\{ \mathbb E[M]: \mathbb E[(Z_k-M)^+]\leq
\alpha_k,\,k=1,\dots,n\} = \mathbb E[Z_n] - \alpha_n.
$$
\item Removing all the constraints except the last one, it is easy to
  see that 
$$
\inf_{M\in \mathcal F}\{ \mathbb E[M]: \mathbb E[(Z_k-M)^+]\leq
\alpha_k,\,k=1,\dots,n\} \geq \mathbb E[Z_n] - \alpha_n.
$$
To finish the proof, it is then enough to find $M$ which satisfies the
constraints and is such that $\mathbb E[M] = \mathbb E[Z_n] - \alpha_n$.
\item If $\mathbb E[Z_n - Z_1]\geq \alpha_n$, we let 
$$
k^* = \max\{i: \mathbb E[Z_n - Z_i] \geq \alpha_n\}
$$ 
and 
$$
M:= w Z_{k^*} + (1-w) Z_{k^*+1},\quad w = \frac{\mathbb E[Z_{k^*+1}] -
  \mathbb E[Z_n] +
  \alpha_n}{\mathbb E[Z_{k^*+1}] - \mathbb E[Z_{k^*}]} \in [0,1]. 
$$
Otherwise, let $k^* = 0$ and
$$
M := Z_1 - C,\quad C = \alpha_n - \mathbb E[Z_n-Z_1]. 
$$
By construction,
$
\mathbb E[M] = \mathbb E[Z_n] - \alpha_n
$
and since $M \leq Z_n$, also $\mathbb E[(Z_n - M)^+] =
\alpha_n$. To check the other constraints, observe that if
$k\leq k^*$ then $M\geq Z_k$ a.s. and so $\mathbb E[(Z_k - M)^+] =
0$. On the other hand, if $k>k^*$ then $M\leq Z_k$ a.s. and so  by \eqref{right}
$$\mathbb E[(Z_k-M)^+] = \mathbb E[Z_k] - \mathbb E[M] = \mathbb E[Z_k]
- \mathbb E[Z_n] + \alpha_n < \alpha_k.$$
\end{enumerate}
\end{proof}


\begin{thebibliography}{10}

\bibitem{bouchard2009stochastic}
{\sc B.~Bouchard, R.~Elie, and N.~Touzi}, {\em Stochastic target problems with
  controlled loss}, SIAM Journal on Control and Optimization, 48 (2009),
  pp.~3123--3150.

\bibitem{bouchard2010optimal}
{\sc B.~Bouchard, L.~Moreau, and M.~Nutz}, {\em Stochastic target games with
  controlled loss}, Annals of Applied Probability, to appear.

\bibitem{bouchard2012stochastic}
{\sc B.~Bouchard and T.~N. Vu}, {\em A stochastic target approach for {P\&L}
  matching problems}, Mathematics of Operations Research, 37 (2012),
  pp.~526--558.

\bibitem{boyle.tian.07}
{\sc P.~Boyle and W.~Tian}, {\em {Portfolio management with constraints}},
  Mathematical Finance, 17 (2007), pp.~319--343.

\bibitem{rcdc.12}
{\sc {Cour de Comptes (French public auditor)}}, {\em Les co\^uts de la
  fili\`ere \'electronucl\'eaire}.
\newblock Public report, January 2012.

\bibitem{detemple2008dynamic}
{\sc J.~Detemple and M.~Rindisbacher}, {\em Dynamic asset liability management
  with tolerance for limited shortfalls}, Insurance: Mathematics and Economics,
  43 (2008), pp.~281--294.

\bibitem{doob}
{\sc J.~L. Doob}, {\em Measure theory}, vol.~143 of Graduate Texts in
  Mathematics, Springer-Verlag, New York, 1994.

\bibitem{elkaroui.al.05}
{\sc N.~El~Karoui, M.~Jeanblanc, and V.~Lacoste}, {\em {Optimal portfolio
  management with American capital guarantee}}, Journal of Economic Dynamics
  and Control, 29 (2005), pp.~449--468.

\bibitem{follmer.leukert.99}
{\sc H.~F\"ollmer and P.~Leukert}, {\em {Quantile hedging}}, Finance and
  Stochastics, 3 (1999), pp.~251--273.

\bibitem{follmerleukert}
{\sc H.~F\"ollmer and P.~Leukert}, {\em Efficient hedging: cost vs. shortfall
  risk}, Finance and Stochastics, 4 (2000), pp.~117--146.

\bibitem{follmer}
{\sc H.~F\"ollmer and A.~Schied}, {\em Stochastic {F}inance}, De Gruyter,
  Berlin, 2002.

\bibitem{gundel.weber.07}
{\sc A.~Gundel and S.~Weber}, {\em {Robust utility maximization with limited
  downside risk in incomplete markets}}, Stochastic Processes and Their
  Applications, 117 (2007), pp.~1663--1688.

\bibitem{kramkov1999asymptotic}
{\sc D.~Kramkov and W.~Schachermayer}, {\em The asymptotic elasticity of
  utility functions and optimal investment in incomplete markets}, Annals of
  Applied Probability, 9 (1999), pp.~904--950.

\bibitem{martellini2012dynamic}
{\sc L.~Martellini and V.~Milhau}, {\em Dynamic allocation decisions in the
  presence of funding ratio constraints}, Journal of Pension Economics and
  Finance, 11 (2012), p.~549.

\bibitem{van2007optimal}
{\sc J.~H. Van~Binsbergen and M.~W. Brandt}, {\em Optimal asset allocation in
  asset liability management}, tech. rep., National Bureau of Economic
  Research, 2007.

\end{thebibliography}
\end{document}